\documentclass[lettersize,journal]{IEEEtran}
\usepackage{amsmath,amsfonts}
\usepackage{algorithmic}
\usepackage{array}
\usepackage[caption=false,font=normalsize,labelfont=sf,textfont=sf]{subfig}
\usepackage{textcomp}
\usepackage{stfloats}
\usepackage{url}
\usepackage{verbatim}
\usepackage{graphicx}
\graphicspath{{figure/}}
\usepackage{bm}
\def\BibTeX{{\rm B\kern-.05em{\sc i\kern-.025em b}\kern-.08em
    T\kern-.1667em\lower.7ex\hbox{E}\kern-.125emX}}
\usepackage{balance}
\usepackage{amsthm}
\usepackage{hyperref}
 \usepackage{float}

\usepackage{booktabs}
\usepackage{graphicx} 
\newtheorem{theorem}{Theorem}[section]
\newtheorem{lemma}[theorem]{Lemma}
\newtheorem{definition}[theorem]{Definition}

\newtheorem{proposition}[theorem]{Proposition}

\newtheorem{remark}[]{Remark}

\newtheorem{example}{Example}

\usepackage[dvipsnames]{xcolor}
\makeatletter
\renewcommand{\textcolor}[2]{%
  \def\tempa{#1}%
  \def\tempb{catalogueblue}%
  \ifx\tempa\tempb
    #2
  \else
    \textcolorOriginal{#1}{#2}
  \fi
}
\makeatother
\begin{document}
\title{\textcolor{catalogueblue}{Quantitative Parameter Conditions for Stability and Coupling in GFM–GFL Converter Hybrid Systems from a Small-Signal Synchronous Perspective}}
\author{Kehao Zhuang, Huanhai Xin, Hangyu Chen, and Linbin Huang  \vspace{-5mm}

}

\maketitle

\begin{abstract}
With the development of renewable energy sources, power systems are gradually evolving into a system comprising both grid-forming (GFM) and grid-following (GFL) converters. However, the dynamic interaction between the two types of converters, especially low-inertia GFM converters and GFL converters, remains unclear due to the substantial differences in their synchronization mechanisms. To address this gap, this paper develops a small-signal synchronous stability model for power systems containing GFM and GFL converters, which considers network line dynamics. \textcolor{catalogueblue}{Based on subspace perturbation theory, we reveal that GFM and GFL subsystems can be effectively decoupled when GFL converters operate near unity power factor or when GFM converters possess sufficiently large inertia or damping, and provide lower bound of control parameters ensuring decoupling.} Under the decoupling condition, we propose decentralized and analytical  parameter-based stability criteria which have clear physical interpretations: the positive damping of converters compensates for the negative damping of the network. In the case of coupling, we also propose decentralized stability criteria based on the small phase theorem. The effectiveness of the theoretical analysis is validated through simulations in MATLAB/Simulink.

\end{abstract}


\section{Introduction}
\IEEEPARstart{S}{ynchronous} stability is a basic requirement for AC power systems operation, which means that all devices connected into the AC grid have same operating frequency \cite{xiongfei:syn_overview}. In conventional power systems, the synchronization is dominated by the rotor angle swing of synchronous generator (SG), which has a clear physical interpretation. Under power imbalance, SGs adjust their rotor speeds to regulate the voltage frequency until power balance is restored and all SGs operate at the same frequency \cite{kunder:stability_classification}. However, with the increasing renewable energy sources integrated into the AC grid through power electronics (PE) converters, the synchronous characteristics of power systems are transformed from physical synchronization dominated by SGs to control synchronization dominated by converters \cite{xiongfei:syn_overview}-\cite{Linbin:PLL_sy}. Unlike SGs, converters adopt more different synchronization methods and exhibit faster synchronous dynamics, which are deeply coupled with each other and network dynamics \cite{Linbin:PLL_sy}. The shift in synchronization mechanisms and more complex dynamics have rendered conventional analysis methods based on SGs no longer applicable. Therefore, it is essential to revisit the synchronous stability to ensure the security of converter-dominated power systems.

In recent decades, the most widely used synchronization control in converters is Phase-Locked Loop (PLL), which rapidly tracks the grid voltage frequency and is thus known as Grid-Following (GFL) control \cite{Thomas:GFM_GFL_comparison}. The increasing GFL converters connected to the grid leads to weak grid characteristics with low short circuit ratio (SCR) and various new synchronous stability issues \cite{Huanhai:gSCR}. For example, many researches have demonstrated that GFL converters will occur sub/super synchronous oscillations dominated by PLL in low SCR weak grid \cite{Linbin:PLL_sy}-\cite{Huanhai:gSCR}. In order to reveal the instability mechanisms of GFL converters, many excellent analysis methods, such as sequence impedance \cite{Sunjian:squence_model,Zhangchen:squence_model}, dq frame impedance \cite{Leon:dq_model,WenBo:dq_model}, polar frame impedance \cite{Huanhai:gimpendece}, and compound torque coefficient \cite{Huqi:torque_method}, were proposed. \textcolor{catalogueblue}{The study shows that the bandwidth of the PLL is a key factor affecting the small-signal synchronization stability of GFL inverters~\cite{iecon:PLL_band,zzx:2021}. The main reason for the small-signal synchronous instability of a GFL converter connected to a grid is the negative damping caused by dynamics coupling between PLL and network line with low SCR. Moreover, when interactions among multiple GFL converters are considered, the damping characteristics of the PLL may be further degraded, potentially leading to oscillations over a wider range~\cite{wgz:numuricalrange}.}

\textcolor{catalogueblue}{Placing Grid-Forming (GFM) converters with large inertia coefficients has been proven to be a promising solution to suppress GFL converter oscillations by enhancing SCR~\cite{Huanhai:place_GFM}-\cite{gfm100}. However, because the core idea of GFM converters is to adopt power synchronization control (PSC, such as virtual synchronous generators, VSGs) by mimicking the SG rotor~\cite{wxf:gfm_review,GFM:keytechonology}, they also inherit the low-frequency oscillation characteristics of SGs, which essentially corresponds to the synchronization instability of GFMs~\cite{dcsygfm:2023,frade:dampeenhance}. Specifically, if low-inertia GFM control is implemented on device such as photovoltaic systems due to limited energy, oscillations caused by GFM synchronization instability may extend from the low-frequency range into the sub- and super-synchronous frequency bands~\cite{cyd_poweroscillation:2025ps}. This can lead to coupling with GFL converters, which may not only fail to suppress GFL-induced oscillations but could also introduce new challenges due to the exacerbating effects of coupling~\cite{cyd_poweroscillation:2025ps,kehao:dual_axis}. }

Although more and more studies focus on the dynamic interaction analysis of GFL and GFM converters, the effect of coupling on the small-signal synchronous stability is not clear. \textcolor{catalogueblue}{The most common approach is to establish impedance models for multiple GFMs and GFLs, and then analyze system stability using the generalized nyquist criterion (GNC) or root locus methods by sweeping through parameter ranges to obtain qualitative stability conclusions \cite{PLL_GFM:whwxf,unifiedmodel:EPSR}. Such modeling couples the dynamics of GFMs, GFLs, and the network together, making it difficult to intuitively reveal the underlying stability mechanisms or derive quantitative parameter conditions. }

\textcolor{catalogueblue}{Furthermore, Ref. \cite{Liyitong:synchronization}\cite{Yangziqian:synchronization} , based on electromechanical models, reveal the quantitative characteristics of GFM–GFL converters coupling at the electromechanical timescale. However, since network dynamics are not considered, these models fail to capture the sub-synchronous oscillations arising from the coupling between converters and the network~\cite{MARUI:multi_time}. }

\textcolor{catalogueblue}{In addition, Ref. \cite{Linbin:gain_phase} proposes a quantization stability method considering network dynamics based on small gain-phase theorem. Unfortunately, it still fails to establish an analytical relationship between control parameters and small-signal synchronization stability, nor does it provide quantitative boundary conditions for the coupling between GFM and GFL converters. To help engineers better understand and plan the system, it is essential to provide simple parameter conditions for small-signal synchronous stability. }



Therefore, the above gaps lead to two unresolved key questions: 1. What are the boundary conditions for synchronous dynamics coupling of GFM and GFL converters? 2. What are the parameter conditions for small signal synchronization stability of GFM and GFL converter hybrid systems? To answer the two questions, this paper presents a synchronization analysis model and a quantitative small-signal stability analysis method. The main four contributions are as follows.

\textcolor{catalogueblue}{1) We establish a synchronous stability analysis model for hybrid GFM–GFL power systems that incorporates network line dynamics. The proposed small-signal model can be simplified to a form similar to a multi-SG second-order model, even at the electromagnetic transient timescale, facilitating intuitive understanding of synchronization stability through concepts such as synchronizing and damping torques.}

\textcolor{catalogueblue}{2) Using the Davis–Kahan subspace perturbation theory, we derive quantitative decoupling conditions for GFM and GFL converters. These conditions provide theoretical lower bounds for GFM inertia and damping, ensuring decoupling from GFLs and helping suppress GFL-induced oscillations. We recommend configuring GFMs to remain decoupled based on these conditions to avoid unintended deterioration.}

\textcolor{catalogueblue}{3) When decoupling is achieved, we propose quantitative parameter conditions for small-signal synchronous stability of the GFL and GFM subsystems.  These conditions separate network information from converter parameters, requiring only a static matrix based on steady-state network flow, and avoid the complexity of computing dynamic frequency-domain matrices for both network and converters.}

\textcolor{catalogueblue}{4) For unintended GFM–GFL coupling, parameter conditions based on the small phase theorem are provided. While these involve the network’s dynamic frequency-domain matrices, they still avoid computing the converters’ dynamic matrices and allow direct determination of controller parameters.  Although more general and applicable to both coupled and decoupled scenarios, using the simple decoupled parameter conditions from 2) and 3) remains more convenient for engineering practice.}

\vspace{-5mm}
\section{Small-Signal Modeling of Multi-Converter Power Systems} 
\textcolor{catalogueblue}{$\theta$ and $\theta_{\rm g}$ are the phases of the local dq frame and the grid voltage, respectively; $\omega$ and $\omega_0$ are the angular frequencies of local $dq$ frame and global $xy$ frame, respectively. $i$ and $v$ denote the current and voltage of converter. The subscripts $_d$ and $_q$ denote the $d$- and $q$-axis components in the local $dq$ frame, The subscripts $_x$ and $_y$ denote the $x$- and $y$-axis components in the global $xy$ frame. $(\cdot)_i$ represents the variable of the $i$-th converter}

\textcolor{catalogueblue}{In this paper, we assume a three-phase balanced power system for analysis. More complex scenarios, such as three-phase unbalanced conditions, require separate and further investigation.}
\subsection{Converters model}

The typical control schemes of a GFL converter and a GFM converter are shown in Fig.\ref{fig1}. A converter synchronizes with the grid via PSC (for GFM) or PLL (for GFL), each generating a local dq frame used for voltage/current vector control. As shown in Fig.\ref{fig2}, the synchronization mechanism of a converter is to ensure the converging phase difference between its local dq frame and the grid voltage by dynamically adjusting its local frame in response to unbalanced synchronous signal, i.e. $\lim \limits_{t \to \infty} \theta-\theta_{\rm g}=\theta_0 $ and $\omega=\omega_0$ ($\theta_0$ is the constant for the steady state phase difference). 
\begin{figure}
	\centering
	\includegraphics[width=3in]{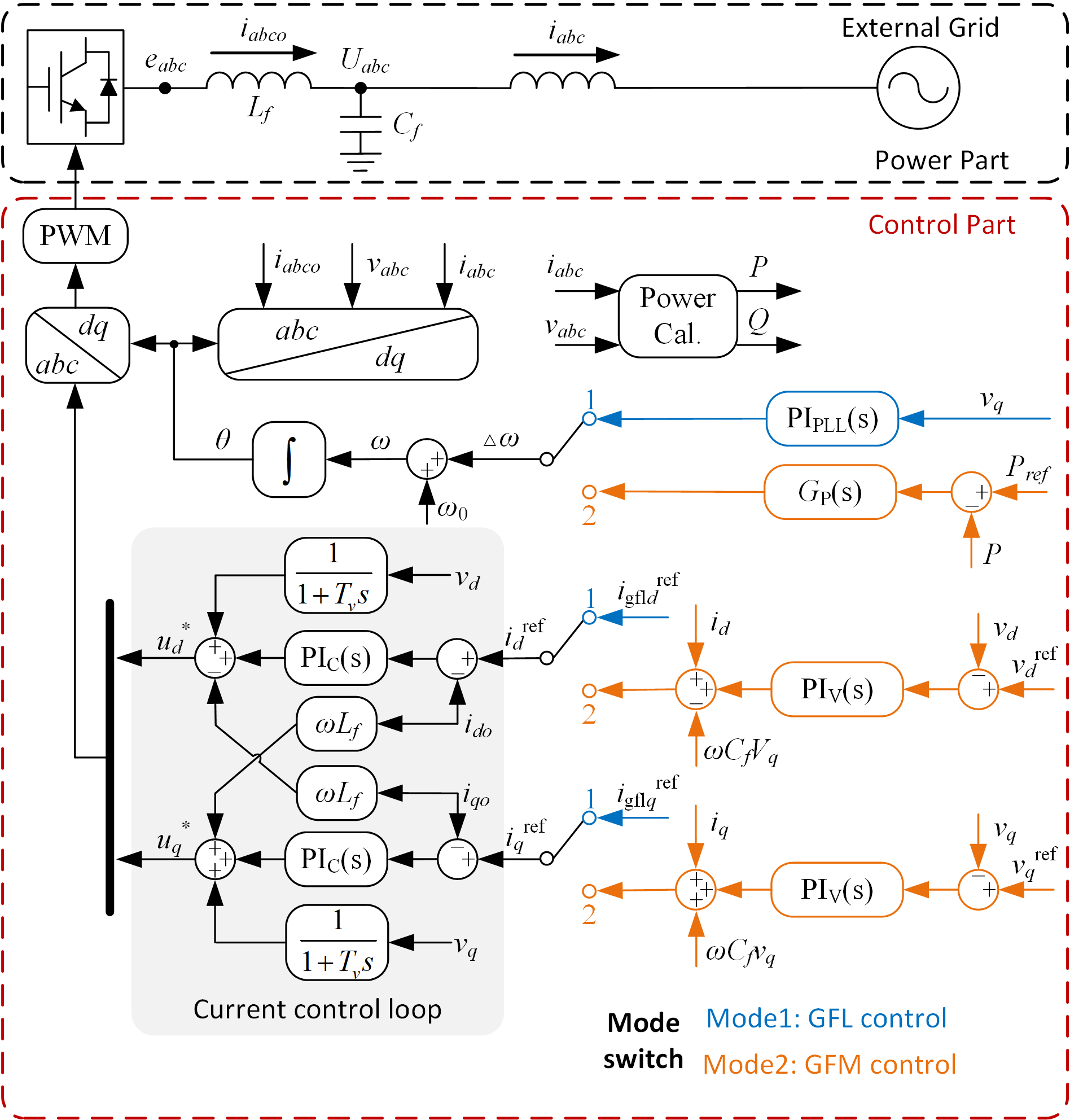}
	\vspace{-3mm}
	\caption{GFM and GFL control schemes of a converter connected to the grid} 
	\vspace{-0.4cm}
	\label{fig1}
\end{figure}
\begin{figure}
	\centering
	\includegraphics[width=1.5in]{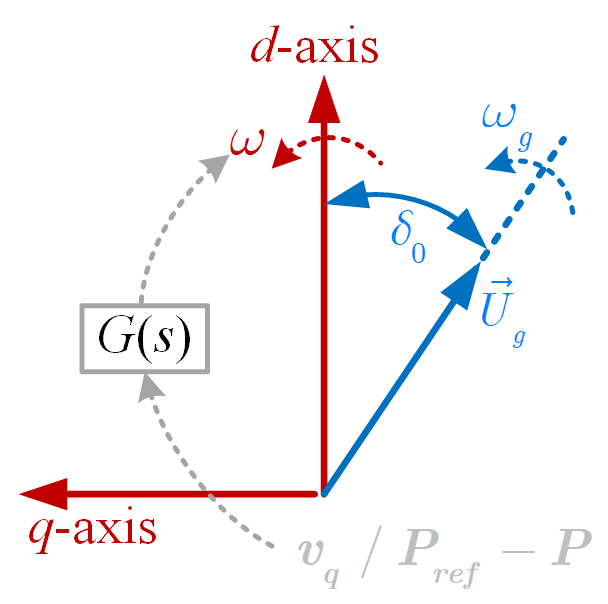}
	\vspace{-3mm}
	\caption{Vector diagram of a converter connected to the grid} 
	\vspace{-0.4cm}
	\label{fig2}
\end{figure}
 
A key difference between the GFL control and the GFM control lies in their synchronous signal. For a GFL converter, its synchronous signal is the voltage q-axis component $v_q$. The $i$-th GFL converter's synchronous control dynamic in small-signal can be expressed as
\vspace{0mm}
\begin{equation}\label{eq:PLL}
\Delta\theta_i=\frac{k_{P,i}s+k_{I,i}}{s^2} \Delta v_{q,i}\,,
\end{equation}
where, $\Delta$ represents the perturbations in small signal, $ {\rm PI}_{\rm PLL}(s)=k_{P,i}+k_{I,i}/s $ is the PI controller of the PLL.

Active power $P$ is the synchronous signal of a GFM converter and its synchronous control model is
\vspace{0mm}
\begin{equation}\label{eq:VSG}
\Delta\theta_i=-\frac{\omega_0}{J_is^2+D_is} \Delta P_{i}\,,
\end{equation}
where, $\omega_0$ is the rated angular frequency, $J_i$ and $D_i$ are the inertial and damping coefficients of the $i$-th converter, respectively.\textcolor{catalogueblue}{ The dynamics of large-capacity SGs are slower and can be regarded as rigid voltage sources. When considering the dynamics of small-capacity SGs, their models can also be approximated by \eqref{eq:VSG} and analyzed in the same way as GFM converters. This case is not separately considered in this paper.}

In addition, another difference between the GFM and GFL \textcolor{catalogueblue}{controls lies} in their different vector control objective, i.e. $dq$-axis voltage $v_{dq}$ for GFM and $dq$-axis current $i_{dq}$ for GFL. Generally, the bandwidth of the synchronous control is narrower than that of the current or voltage vector control loop. In this situation, the GFL converter's current and the GFM converter's voltage are approximated by their reference values in synchronous stability analysis, i.e. $i_d=i_{dref}={\rm constant}$, $i_q=i_{qref}={\rm constant}$ for the GFL converter and  $v_d=v_{dref}={\rm constant}$, $v_q=v_{qref}=0$ for the GFM converter. 

Therefore, the active power of GFM converters can be written as $\Delta P=v_d\Delta i_d$ and $v_d$ can be considered as the constant control gain. We take a unified formula to describe the synchronous dynamic of GFM and GFL converters, \textcolor{catalogueblue}{
\vspace{0mm}
\begin{equation}\label{eq:unified syn}
\Delta\theta_i=-\frac{T_{P,i}s+1}{T_{J,i}s^2+T_{D,i}s} \Delta M_{i}=:-G_i(s)\Delta M_i\,,
\end{equation}
where, for a GFL converter, $T_{P,i}=1/k_{I,i}, T_{J,i}=0, T_{D,i}=k_{P,i}/k_{I,i}, \Delta M_i=-\Delta v_{q,i}$. For a GFM converter, $T_{J,i}=J_i/v_{d,i}\omega_0, T_{D,i}=D_i/v_{d,i}\omega_0, T_{P,i}=0, \Delta M_i=\Delta i_{d,i}$.}

\begin{figure}
	\centering
	\includegraphics[width=3in]{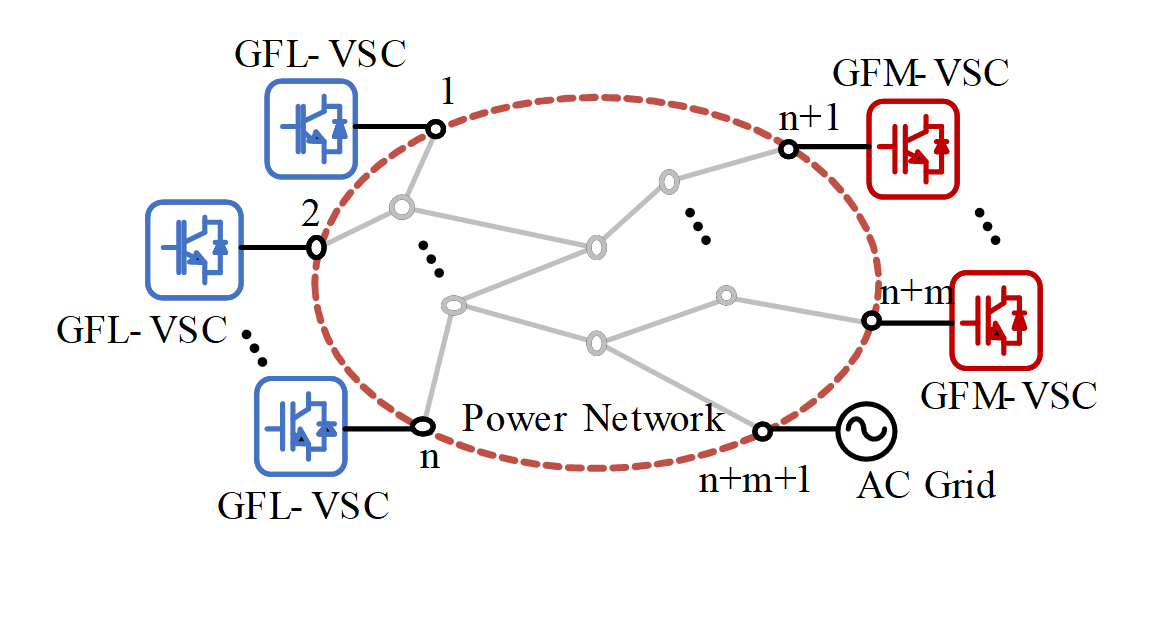}
	\vspace{-3mm}
	\caption{GFM and GFL hybrid power system} 
	\vspace{-0.4cm}
	\label{fig3}
\end{figure}

We consider a power system consisting of $n$ GFL converters and $m$ GFM converters, which is shown in Fig. \ref{fig3}. The small-signal model of the multi-converter can be written as
\vspace{0mm}
\begin{equation}\label{eq:multiconverter}
\underbrace{\begin{bmatrix}
    \Delta {\bm \theta}_{\rm gfl}  \\ \Delta {\bm \theta}_{\rm gfm}
\end{bmatrix}}_{=:\Delta {\bm \theta}}
= -{\bm G}_{\rm sw}(s)
\underbrace{
\begin{bmatrix}
    -\Delta {\bm V}^{\rm gfl}_{q}  \\ \Delta {\bm I}^{\rm gfm}_{d}
\end{bmatrix}}_{=:\Delta {\bm M}}\,,
\end{equation}
where, 
\[
\begin{aligned}
\Delta {\bm \theta}_{\rm gfl}&=\begin{bmatrix} \Delta \theta_1 & \dots & \Delta \theta_n\end{bmatrix}^\top, \\
\Delta {\bm \theta}_{\rm gfm}&=\begin{bmatrix} \Delta \theta_{n+1} & \dots & \Delta \theta_{n+m}\end{bmatrix}^\top,\\
{\bm G}_{\rm sw}(s)&={\rm diag}\left\{ G_{1}(s), \dots,  G_{n+m}(s) \right\},\\
\Delta {\bm V}^{\rm gfl}_{q}&=\begin{bmatrix}  \Delta v_{q,1} &\dots & \Delta v_{q,n} \end{bmatrix}^\top,\\
\Delta {\bm I}^{\rm gfm}_{d}&=\begin{bmatrix}  \Delta i_{d,n+1} &\dots &\Delta i_{d,n+m} \end{bmatrix}^\top    \,,
\end{aligned}
\]
and ${\rm diag}\left\{\cdot\right\}$ represents a diagonal matrix.

\subsection{Network model}
We denote the voltage and current vectors of converter nodes in global $xy$ frame (angular frequency is $\omega_0$) as $\Delta {\bm V}_{xy}=[\Delta {\bm V}^{\rm gfl \top}_{xy}\; \Delta {\bm V}^{\rm gfm \top}_{xy}]^\top$ and $\Delta {\bm I}_{xy}=[\Delta {\bm I}^{\rm gfl \top}_{xy}\; \Delta {\bm I}^{\rm gfm \top}_{xy}]^\top$ \textcolor{catalogueblue}{(composed of the voltage and current vectors of each converter in the $xy$ frame: $v_{x,i},v_{y,i}$ and $i_{x,i},i_{y,i}$)}, the network dynamics are given by
\vspace{0mm}
\begin{equation}\label{eq:Ygrid}
\begin{bmatrix}  \Delta{\bm I}^{\rm gfl}_{xy}\\ \Delta{\bm I}^{\rm gfm}_{xy}\end{bmatrix}=
\left[ {\bm Y}_{\rm net} \otimes  \gamma(s)\right]
\begin{bmatrix} \Delta{\bm V}^{\rm gfl}_{xy}\\ \Delta{\bm V}^{\rm gfm}_{xy}\end{bmatrix}\,,
\end{equation}
with
\vspace{-0mm}
\begin{equation}\label{eq:gamma}
\begin{aligned}
{\bm Y}_{\rm net}:=&\begin{bmatrix} {\bm Y}_1 \in \mathbb{R}^{n\times n} &  {\bm Y}_2 \in \mathbb{R}^{n\times m}  \\ {\bm Y}_3 \in \mathbb{R}^{m\times n} & {\bm Y}_4\in \mathbb{R}^{m\times m}  \end{bmatrix}, \\
\gamma(s):=&\begin{bmatrix} \tau+s/\omega_0 & -1 \\ 1 & \tau+s/\omega_0  \end{bmatrix}^{-1} \,,    
\end{aligned}
\end{equation}
where ${\bm Y}_{\rm net}$ is the reduced network admittance matrix and its $ij$-th element represents equivalent admittance between node $i$ and node $j$, $\tau$ is the line resistance to inductance ratio ($R/X$), $\otimes$ is the Kronecker product.

GFL converters' current and GFM converters' voltage serve as the inputs to the network, respectively. Then the network dynamics are reorganized into
\vspace{0mm}
\begin{equation}\label{eq:ZY}
\begin{bmatrix} \Delta{\bm V}^{\rm gfl}_{xy}\\ \Delta{\bm I}^{\rm gfm}_{xy}\end{bmatrix}=
{\bm G}_{\rm net}(s)
\begin{bmatrix} \Delta{\bm I}^{\rm gfl}_{xy}\\ \Delta{\bm V}^{\rm gfm}_{xy}\end{bmatrix} \,,
\end{equation}
with
\vspace{0mm}
\begin{equation}\label{eq:ZYdynamic}
{\bm G}_{\rm net}(s):= \begin{bmatrix} {\bm Y}^{-1}_1\otimes  \gamma^{-1}(s) &  -{\bm Y}_1^{-1}{\bm Y}_2\otimes {\bm I}_2 \\ {\bm Y}_3{\bm Y}_1^{-1}\otimes {\bm I}_2 & \left({\bm Y}_4-{\bm Y}_3{\bm Y}_1^{-1}{\bm Y}_2 \right)\otimes  \gamma(s) \end{bmatrix}
 \,,
\end{equation}
where ${\bm I}_n$ is a $n \times n$ identity matrix. The derivation is shown in Appendix \ref{sec:appendixA}. 

The relationship between a vector $C$ in the global $xy$ frame and the local dq frame is
\vspace{0mm}
\begin{equation}\label{eq:dqtoxy}
\begin{bmatrix}
    C_{d,i} \\ C_{q,i}
\end{bmatrix}=
\begin{bmatrix} {\rm cos}\theta_i & {\rm sin}\theta_i \\ -{\rm sin}\theta_i & {\rm cos}\theta_i \end{bmatrix}
\begin{bmatrix}
    C_{x,i} \\ C_{y,i}
\end{bmatrix}\,.
\end{equation}

Then we obtain $-v_{q,i}=\begin{bmatrix} {\rm sin}\theta_i & -{\rm cos}\theta_i \end{bmatrix}\begin{bmatrix} v_{x,i} & v_{y,i} \end{bmatrix}^\top$ and $i_{d,i}=\begin{bmatrix} {\rm cos}\theta_i & {\rm sin}\theta_i \end{bmatrix}\begin{bmatrix} i_{x,i} & i_{y,i} \end{bmatrix}^\top$. Linearize the above equation and transform the network outputs $\Delta{\bm V}^{\rm gfl}_{xy}$ and $\Delta{\bm I}^{\rm gfm}_{xy}$ into $-\Delta {\bm V}^{\rm gfl}_q$ and $\Delta {\bm I}^{\rm gfm}_d$. 
\vspace{0mm}
\begin{equation}\label{eq:output}
\Delta {\bm M}=
\underbrace{ \begin{bmatrix} {\bm T}^{\rm gfl}_{{\rm out}} &  {\bm 0} \\ {\bm 0} &{\bm T}^{\rm gfm}_{{\rm out}} \end{bmatrix}}_{=:{\bm T}_{\rm out}}
\begin{bmatrix} \Delta{\bm V}^{\rm gfl}_{xy}\\ \Delta{\bm I}^{\rm gfm}_{xy}\end{bmatrix}+
\begin{bmatrix} {\bm V}^{\rm gfl}_{d} & {\bm 0}\\ {\bm 0} & {\bm I}^{\rm gfm}_{q}\end{bmatrix}\Delta{\bm \theta}
  \,,
\end{equation}
where ${\bm T}^{\rm gfl}_{{\rm out}}$, ${\bm T}^{\rm gfm}_{{\rm out}}$, $ \bm V^{\rm gfl}_d$ and $\bm I^{\rm gfm}_q$ are all diagonal block matrices. ${\bm T}^{\rm gfl}_{{\rm out},i}=\begin{bmatrix} {\rm sin}\theta_i & -{\rm cos}\theta_i \end{bmatrix}$, ${\bm T}^{\rm gfm}_{{\rm out},i}=\begin{bmatrix}  {\rm cos}\theta_i & {\rm sin}\theta_i\end{bmatrix}$,  $ \bm V^{\rm gfl}_{d,i}=\begin{bmatrix}{\rm cos}\theta_i & {\rm sin}\theta_i\end{bmatrix}\begin{bmatrix} v_{x,i} & v_{y,i} \end{bmatrix}^\top=v_{d,i}$ and $ \bm I^{\rm gfm}_{q,i}=\begin{bmatrix}-{\rm sin}\theta_i & {\rm cos}\theta_i\end{bmatrix}\begin{bmatrix} i_{x,i} & i_{y,i} \end{bmatrix}^\top=i_{q,i}$. 

Because we ignore the current and voltage control dynamics of converters, i.e.  \textcolor{catalogueblue}{$\Delta i_{d}=\Delta i_{q}=0$ and $\Delta v_{d}=\Delta v_{q}=0$}, the linearized input of the network in global $xy$ frame can be expressed as
\vspace{0mm}
\begin{equation}\label{eq:input}
\begin{aligned}
\begin{bmatrix} \Delta{\bm I}^{\rm gfl}_{xy}\\ \Delta{\bm V}^{\rm gfm}_{xy}\end{bmatrix}& = {\bm T}_{\rm in}\begin{bmatrix} {\bm I}^{\rm gfl}_{dq} & {\bm 0}\\ {\bm 0} & {\bm V}^{\rm gfm}_{dq}\end{bmatrix}\Delta{\bm \theta}\,,
\end{aligned}
\end{equation}
where ${\bm T}_{\rm in}$, ${\bm V}^{\rm gfm}_{dq}$ and ${\bm I}^{\rm gfl}_{dq}$ are all the diagonal block matrix. ${\bm T}_{{\rm in},i}=\left[\begin{smallmatrix} -{\rm sin}\theta_i & -{\rm cos}\theta_i \\ {\rm cos}\theta_i & -{\rm sin}\theta_i \end{smallmatrix} \right]$, ${\bm V}^{\rm gfm}_{dq,i}=\begin{bmatrix} v_{d,i}&v_{q,i}  \end{bmatrix}^\top$ and ${\bm I}^{\rm gfl}_{dq,i}=\begin{bmatrix} i_{d,i}&i_{q,i}  \end{bmatrix}^\top$.

Combining \eqref{eq:ZY} to \eqref{eq:input}, the network model can be expressed with $ \Delta \bm \theta $ as input and $\Delta {\bm M}$ as output. 
\vspace{0mm}
\begin{equation}\label{eq:finalnetwork}
\begin{aligned}
\Delta {\bm M}&= \left ({\bm T}_{\rm out}  {\bm G}_{\rm net}(s) {\bm T}_{\rm in}  
\begin{bmatrix} {\bm I}^{\rm gfl}_{dq} & {\bm 0}\\ {\bm 0} & {\bm V}^{\rm gfm}_{dq}\end{bmatrix}+
\begin{bmatrix} {\bm V}^{\rm gfl}_{d} & {\bm 0}\\ {\bm 0} & {\bm I}^{\rm gfm}_{q}\end{bmatrix} \right)
\Delta \bm \theta \\
& =:\left({\bm K}^{\rm sy}_{\rm net}(s)+s{\bm K}^{\rm d}_{\rm net}(s) \right) \Delta \bm \theta
  \,,
\end{aligned}
\end{equation}
with\textcolor{catalogueblue}{
\vspace{0mm}
\begin{equation}\label{eq:netfenjie}
\begin{aligned}
{\bm K}^{\rm sy}_{\rm net}(s)&=
\begin{bmatrix} 
{\bm K}^{\rm sy}_{\rm gfl}+{\bm V}^{\rm gfl}_{d}& {\bm K}^{\rm sy}_{\rm lm}\\
{\bm K}^{\rm sy}_{\rm ml} &  \alpha(s){\bm K}^{\rm sy}_{\rm gfm}+{\bm I}^{\rm gfm}_{q}
\end{bmatrix}\\
&=:\begin{bmatrix} {\bm L}_{1} & {\bm L}_{2}\\ {\bm L}_{3} & {\bm L}_{4}(s) \end{bmatrix}\\
{\bm K}^{\rm d}_{\rm net}(s)&=
\begin{bmatrix} 
{\bm K}^{\rm d}_{\rm gfl}& \bm 0\\
\bm 0  &  \alpha(s){\bm K}^{\rm d}_{\rm gfm}
\end{bmatrix}
  \,,
\end{aligned}
\end{equation}
where, 
\begin{equation}\label{eq:elementofK}
 \begin{aligned}
{\bm K^{\rm sy}_{{\rm gfl},ij}}&=-I_jX_{ij}[\tau {\rm cos}(\theta_{ij}-\varphi_j)+{\rm sin}(\theta_{ij}-\varphi_j)],\\
{\bm K}^{\rm sy}_{{\rm lm},ij}&=-V_{j}b_{ij}{\rm cos}\theta_{ij},\\
{\bm K}^{\rm sy}_{{\rm ml},ij}&=V_{i}I_jc_{ij}{\rm sin}(\theta_{ij}-\varphi_j),\\
{\bm K^{\rm sy}_{{\rm gfm},ij}}&=V_{i}V_{j}Y_{ij}(\tau {\rm sin}\theta_{ij}+{\rm cos}\theta_{ij})\\
{\bm K^{\rm d}_{{\rm gfl},ij}}&=-I_j(X_{ij}/\omega_0) {\rm cos}(\theta_{ij}-\varphi_j),\\
{\bm K^{\rm d}_{{\rm gfm},ij}}&=V_{i}V_{j}Y_{ij} {\rm sin}\theta_{ij}/\omega_0,\\
\alpha(s)&=\frac{1}{(s/\omega_0)^2+2(\tau/\omega_0)s+\tau^2+1}
\,,
\end{aligned}
\end{equation}
and $\varphi_i={\rm arctan}(i_{q,i}/i_{d,i})$ is the power factor angle of the $i$-th GFL converter, $I_i=\sqrt{i^2_{d,i}+i^2_{q,i}}$ is the current magnitude of the $i$-th GFL converter, $V_i$ is the voltage magnitude of the $i$-th converter, $X_{ij}$, $Y_{ij}$, $b_{ij}$ and $c_{ij}$ are the elements of ${\bm Y}_1^{-1}$, ${\bm Y}_4-{\bm Y}_3{\bm Y}_1^{-1}{\bm Y}_2$, $-{\bm Y}^{-1}_1{\bm Y}_2$ and ${\bm Y}_3{\bm Y}^{-1}_1$ in ${\bm G}_{\rm net}(s)$ of \eqref{eq:ZYdynamic}, respectively. }

\subsection{Closed loop system model}
\begin{figure}
	\centering
	\includegraphics[width=2.8in]{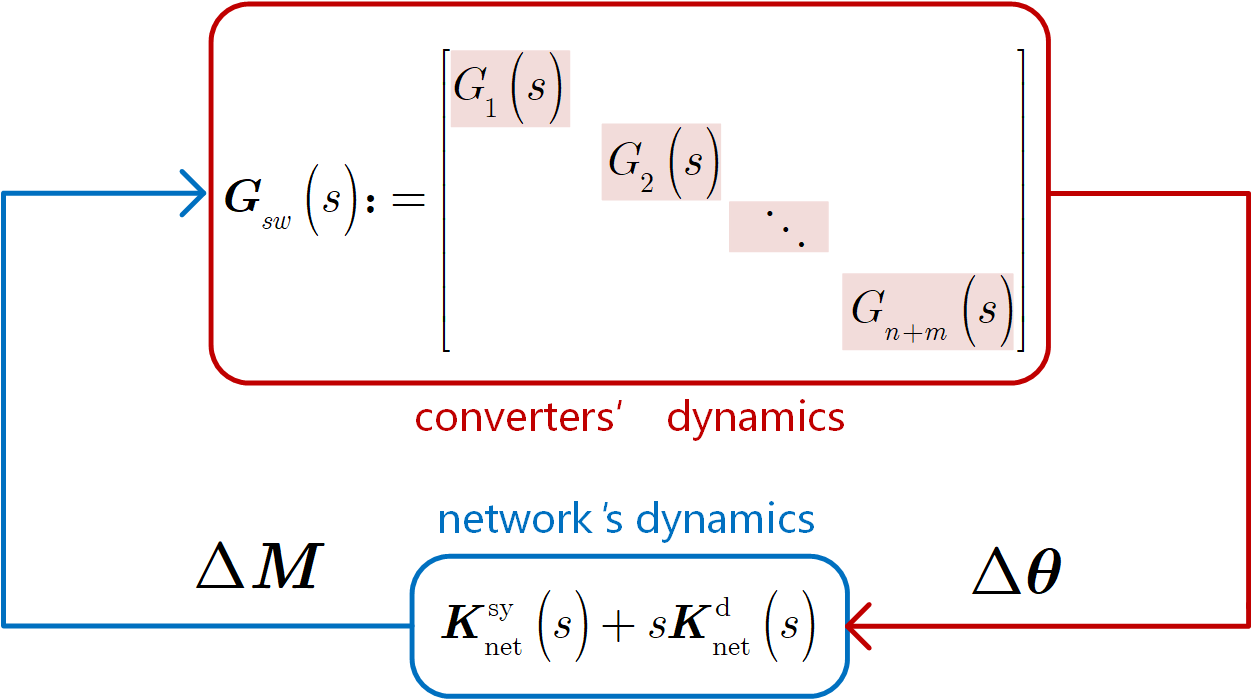}
	\vspace{-3mm}
	\caption{The equivalent control block diagram of the proposed model} 
	\vspace{-0.4cm}
	\label{fig4}
\end{figure}

As shown in Fig. \ref{fig4}, \eqref{eq:multiconverter} and \eqref{eq:finalnetwork} constitute the closed-loop dynamics of the multi-converter power system. The characteristic polynomial is following:
\begin{equation}\label{eq:closedloop}
\begin{aligned}
&{\rm det}\left[{ \bm I}_{n+m}+\bm G_{\rm sw}(s)\left({\bm K}^{\rm sy}_{\rm net}(s)+s{\bm K}^{\rm d}_{\rm net}(s) \right) \right]=0
 \,.    
\end{aligned}
\end{equation}

After linear transformation (left-multiply Eq.\eqref{eq:closedloop} by matrix ${\bm T}_Js^2+{\bm T}_Ds$), the characteristic polynomial is equivalent as 
\begin{equation}\label{eq:2citezheng}
\begin{aligned}
 {\rm det}\left[{\bm H}(s)s^2+{\bm D}(s)s+{\bm L}(s)\right]=0
 \,,    
\end{aligned}
\end{equation}
with,
\begin{equation}\label{eq:HDL}
\begin{aligned}
{\bm H}(s)&= {\bm T}_J+{\bm T}_P{\bm K}^{\rm d}_{\rm net}(s)\approx  {\bm T}_J=:
\begin{bmatrix} {\bm H}_{\rm gfl} & {\bm 0} \\ {\bm 0} & {\bm H}_{\rm gfm}    \end{bmatrix},\\
{\bm D}(s)&={\bm T}_D+{\bm T}_P{\bm K}^{\rm sy}_{\rm net}(s)+{\bm K}^{\rm d}_{\rm net}(s)=:\begin{bmatrix} {\bm D}_{\rm gfl} & {\bm 0} \\ {\bm 0} & {\bm D}_{\rm gfm}(s)   \end{bmatrix},\\
{\bm L}(s)&={\bm K}^{\rm sy}_{\rm net}(s)=\begin{bmatrix} {\bm L}_{1} & {\bm L}_{2}\\ {\bm L}_{3} & {\bm L}_{4}(s)   \end{bmatrix}\,,
\end{aligned}
\end{equation}
where, ${\bm T}_J={\rm diag}\left\{T_{J,i}\right\},\,{\bm T}_D={\rm diag}\left\{T_{D,i}\right\},\,{\bm T}_P={\rm diag}\left\{T_{P,i}\right\},\,i \in [1,n+m]$; ${\bm H}_{\rm gfl}={\rm diag}\left\{T_{J,i}\right\},i\in [1,n]$; ${\bm H}_{\rm gfm}={\rm diag}\left\{T_{J,i}\right\},i\in [n+1,n+m]$; ${\bm D}_{\rm gfl}={\rm diag}\left\{T_{P,i}\right\}({\bm K}^{\rm sy}_{\rm gfl}+{\bm V}^{\rm gfl}_0)+{\bm K}^{\rm d}_{\rm gfl},i\in [1,n]$; ${\bm D}_{\rm gfm}(s)={\rm diag}\left\{T_{D,i}\right\}+\alpha(s){\bm K}^{\rm d}_{\rm gfm},i\in[n+1,n+m]$. \textcolor{catalogueblue}{The reason ${\bm H}(s)$ can be approximated by ${\bm T}_J$ is: ${\bm T}_P{\bm K}^{\rm d}_{\rm net}(s)$ has nonzero values only in the top-left block corresponding to the GFL with each element magnitude less than $(k_{P,i}I_jX_{ij})/(k_{I,i}\omega_0)$, which is much smaller than the elements $1/k_{I,i},J_i/\omega_0v_{d,i}$ in ${\bm T}_J$ due to $(k_{P,i}I_jX_{ij})/\omega_0 \ll  1$.}

The form of \eqref{eq:2citezheng}  resembles the classic model used for synchronous stability analysis of multi-SG systems. ${\bm H}(s)$ is an equivalent inertial matrix, ${\bm D}(s)$ is an equivalent damping matrix, and ${\bm L}(s)$ is an equivalent synchronization coefficient matrix. Therefore, the physical insights and analytical methods developed for SGs can be extended to the multi-converter power system. ${\bm L}(s)$ and ${\bm D}(s)$ mainly provide equivalent ``{\em synchronous torque}" and ``{\em damping torque}", respectively. ${\bm H}(s)$ provides the inertia time constant, which mainly influences the synchronous dynamic response speed, and to some extent affects whether GFM and GFL converters are coupled on the same time scale. \textcolor{catalogueblue}{In addition, it is worth noting that if the network is modeled using the electromechanical approach, ${\bm K}^{\rm d}_{\rm net}(s)=0$. This completely neglects the negative damping torque caused by the network and renders the results inaccurate.} To verify the validity of the proposed model \eqref{eq:2citezheng}, the system shown in Fig.\ref{fig5} is used for testing, the line parameters can refer to \cite{ieee39}.

\begin{figure} 
	\centering
	\includegraphics[width=3in]{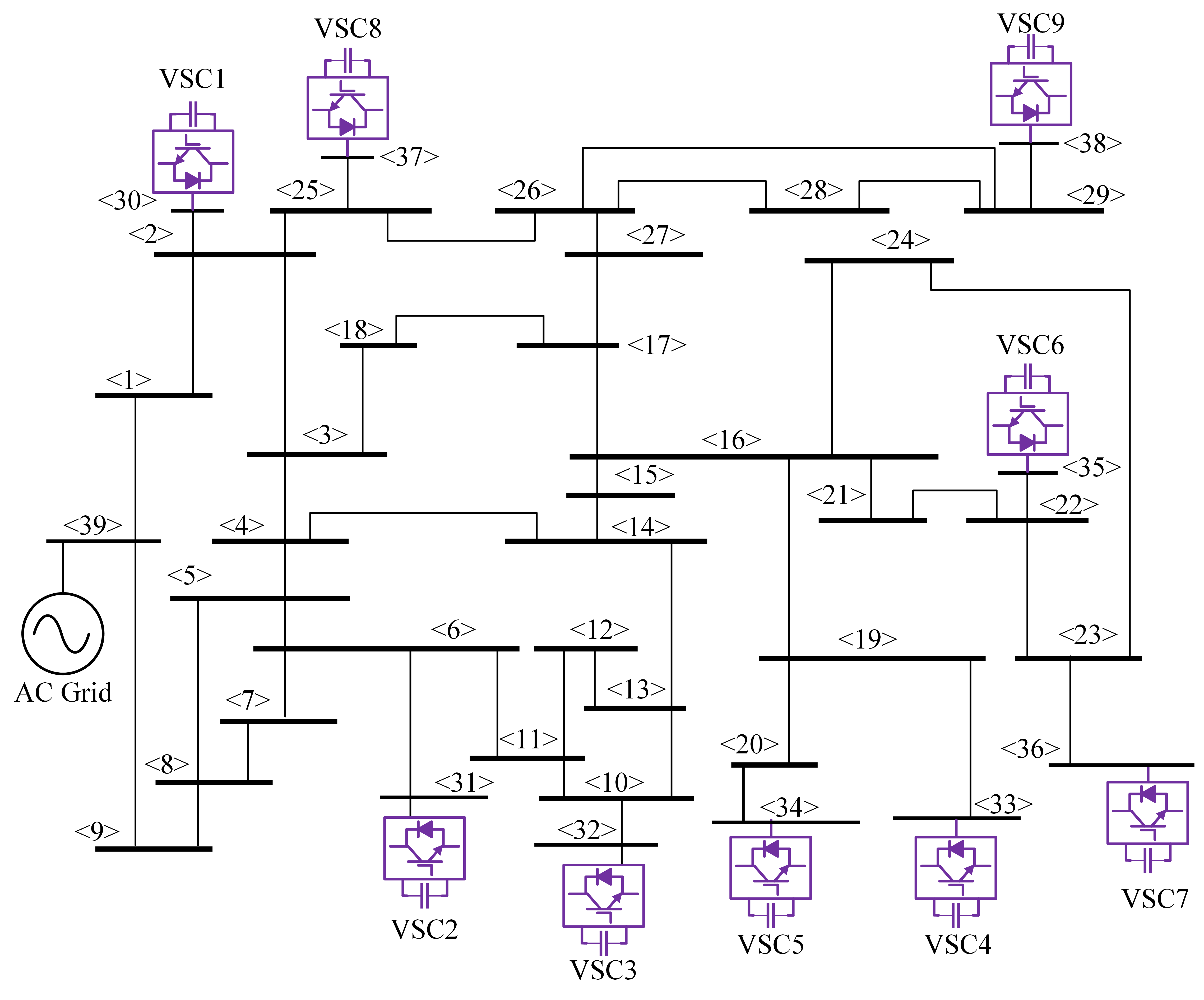}
	\vspace{-3mm}
	\caption{IEEE 39 buses system} 
	\vspace{-0.4cm}
	\label{fig5}
\end{figure}
\vspace{-0mm}
\begin{example}
In Fig.\ref{fig5}, VSC1-VSC6 are GFL converters, VSC7-VSC9 are GFM converters, the parameters are shown in Appendix \ref{sec:appendixE}. As shown in Fig.\ref{fig6}, we compare the eigenvalues of the proposed model ignoring voltage and current control, with those of the full-order model, under different damping coefficients of GFM converters. The top left displays all eigenvalues of both the full-order and proposed models, while the top right highlights the dominant modes, which are primarily dominated by synchronous dynamics. By reducing the damping coefficients of the GFM converters, the eigenvalues of the dominant modes are shown in the bottom left and bottom right of Fig.\ref{fig6}. The comparison exhibits close agreement between the proposed and full-order models. Similar trends in variation and small errors confirm the accuracy and validity of the proposed model. 
\end{example}
\vspace{-5mm}
\begin{figure} 
	\centering
	\includegraphics[width=2.8in]{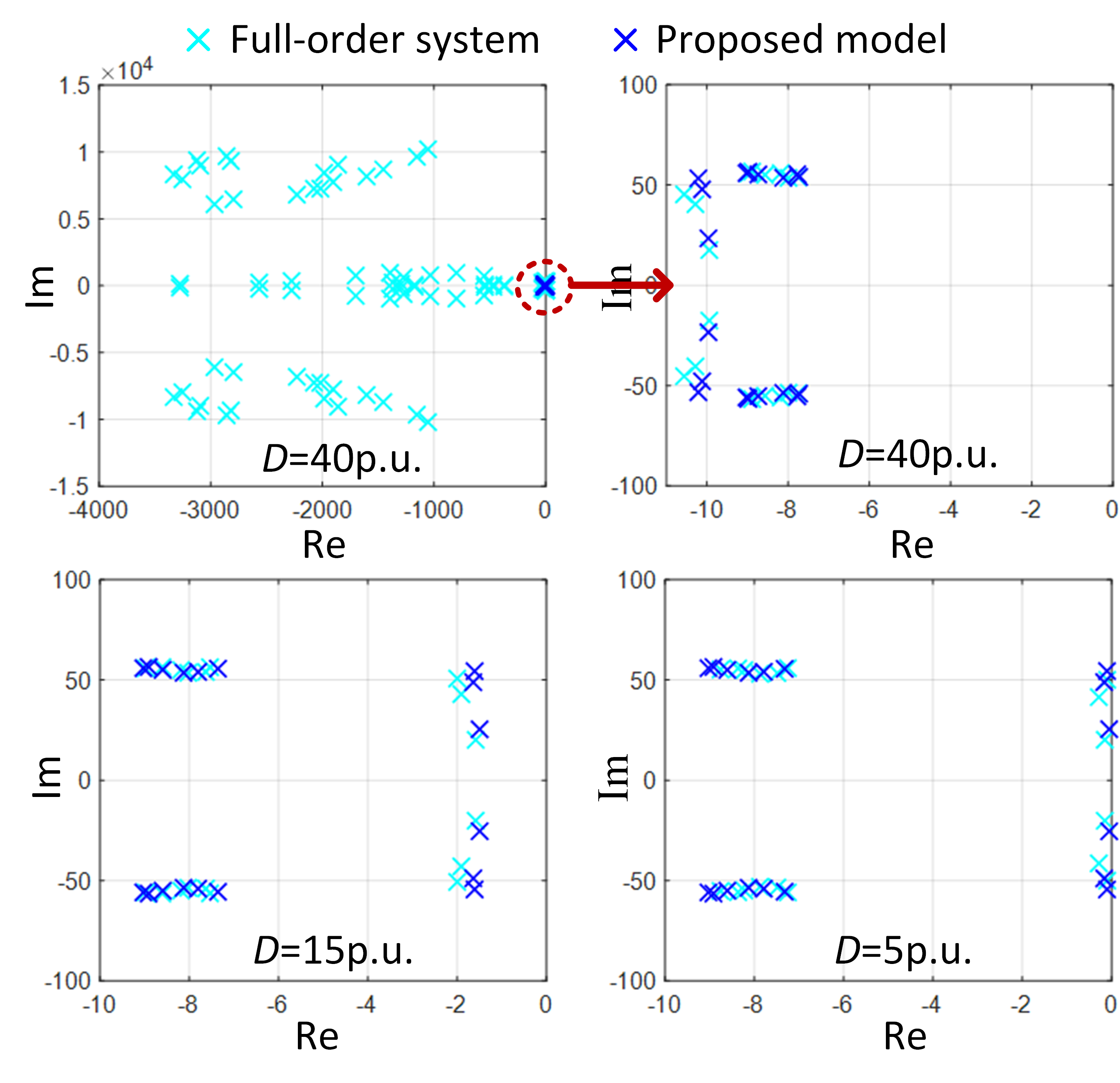}
	\vspace{-3mm}
	\caption{The eigenvalues of the proposed model and the full-order model} 
	\vspace{-0.4cm}
	\label{fig6}
\end{figure}
\section{Decoupling Analysis of GFM and GFL Converters} 
\subsection{The decoupling conditions of GFM and GFL converters}
\textcolor{catalogueblue}{
Observing~\eqref{eq:HDL}, we find that ${\bm H}(s)$ and ${\bm D}(s)$ are diagonal matrices, while only ${\bm L}(s)$ is a full matrix. Therefore, if the off-diagonal elements of ${\bm L}(s)$ have little impact on the overall matrix, the GFL and GFM converters can be analyzed separately. We will next introduce subspace perturbation theory to derive quantitative conditions for decoupling the GFL and GFM converters. }

\textcolor{catalogueblue}{
{\bf Davis-Kahan subspace perturbation theory~\cite{kato1966perturbation}:} Let $A={\rm diag}(A_1,A_2)\in \mathbb{C}$, and let $E={\rm antidiag}(E_1,E_2)\in \mathbb{C}$ be a perturbation (antidiag denotes off-diagonal block matrix). The eigenvalues of $A$ be partitioned into two disjoint sets $\lambda(A)=\lambda(A_1)\cup\lambda(A_2)$, with $\min|{\lambda(A_1)-\lambda(A_2)}|=\delta>0$. Let $V$ be the invariant subspaces corresponding to $\lambda(A)$, and let $W$ be the invariant subspace of $\lambda(A+E)$, the largest principal angle between $V$ and $W$ is
\begin{equation}\label{eq:subspace}
\begin{aligned}
\|\sin\Theta(V,W)\|_2 \leq \kappa \frac{\|E_1\|_2\|E_2\|_2}{\delta}\,,
\end{aligned}
\end{equation}
where, $\kappa$ is the condition number of the eigenvector matrix of $A$, $\|\|_2$ is the $\ell_2$ norm.}

\textcolor{catalogueblue}{If the subspace remains nearly unchanged ($\|\sin\Theta(V,W)\|_2 \leq\epsilon$, $\epsilon$ is a small constant), it implies that $E$ does not affect the eigenvalues of $A$, and $A+E$ can still be approximately decoupled as $A={\rm diag}(A_1,A_2)$.}

\begin{proposition}[{\bf The decoupling conditions}]\label{The decoupling condition}
\textcolor{catalogueblue}{The \eqref{eq:2citezheng} can be decoupled as two subsystem
\begin{equation}\label{eq:decoupled}
\begin{aligned}
\left\{\begin{matrix}
 {\rm det}[{\bm H}_{\rm gfl}s^2+{\bm D}_{\rm gfl}s+{\bm L}_1]=: {\rm det}[{\bm N}_{\rm gfl}(s)]=0\\
 {\rm det}[{\bm H}_{\rm gfm}s^2+{\bm D}_{\rm gfm}(s)s+{\bm L}_4(s)]=:{\rm det}[{\bm N}_{\rm gfm}(s)]=0
\end{matrix}\right.
\,,
\end{aligned}
\end{equation}
if \eqref{eq:decoupled condition} holds,
\begin{equation}\label{eq:decoupled condition}
\begin{aligned}
&\frac{\|{\bm L}_2\|_2\|{\bm L}_3\|_2}{\delta}\leq\epsilon,\\
\delta^2=&\left(\underset{\forall i,j}{\min}\left|\frac{J_i\omega^2}{\omega_0}-\lambda_j(L_4)\right|-\underset{\forall i,j}{\max}\left|\frac{\omega^2}{k_{I,i}}-\lambda_j(L_1)\right|\right)^2\\
&+\underset{\forall i,j}{\min}\left|\frac{D_i}{\omega_0}-\frac{k_{P,j}}{k_{I,j}}\right|^2\omega^2\,,\omega\in[\omega_{\min},\omega_{\max}]
\end{aligned}
\end{equation}
where, $\epsilon$ is a small constant, $\left[\omega_{\min},\omega_{\max} \right]$ represents the PLL bandwidth range of the GFL converters. For ideal decoupling, we should set $\epsilon=0$. However, in practice, setting $\epsilon=0.05$ is sufficient to ensure a decoupling accuracy. }

\textcolor{catalogueblue}{If a simpler parameter condition is desired, \eqref{eq:decoupled condition} can be further simplified, requiring only the design of a lower bound for the damping coefficient to satisfy \eqref{eq:simple condition} to ensure decoupling (more conservative but very simple),
\begin{equation}\label{eq:simple condition}
\begin{aligned}
D_{i}\geq \frac{\|{\bm L}_2\|_2\|{\bm L}_3\|_2\omega_0}{\epsilon\omega^2_{\min}}+\max\left(\frac{k_{P,j}}{k_{I,j}}\right)\omega_0, \forall i,j\,,
\end{aligned}
\end{equation}
where $\omega_{min}$ is the smallest PLL bandwidth of GFL converter. }
\end{proposition}

\begin{proof}
\textcolor{catalogueblue}{See Appendix~\ref{sec:appendixB}}
\end{proof}

\begin{figure}
	\centering
	\includegraphics[width=2.5in]{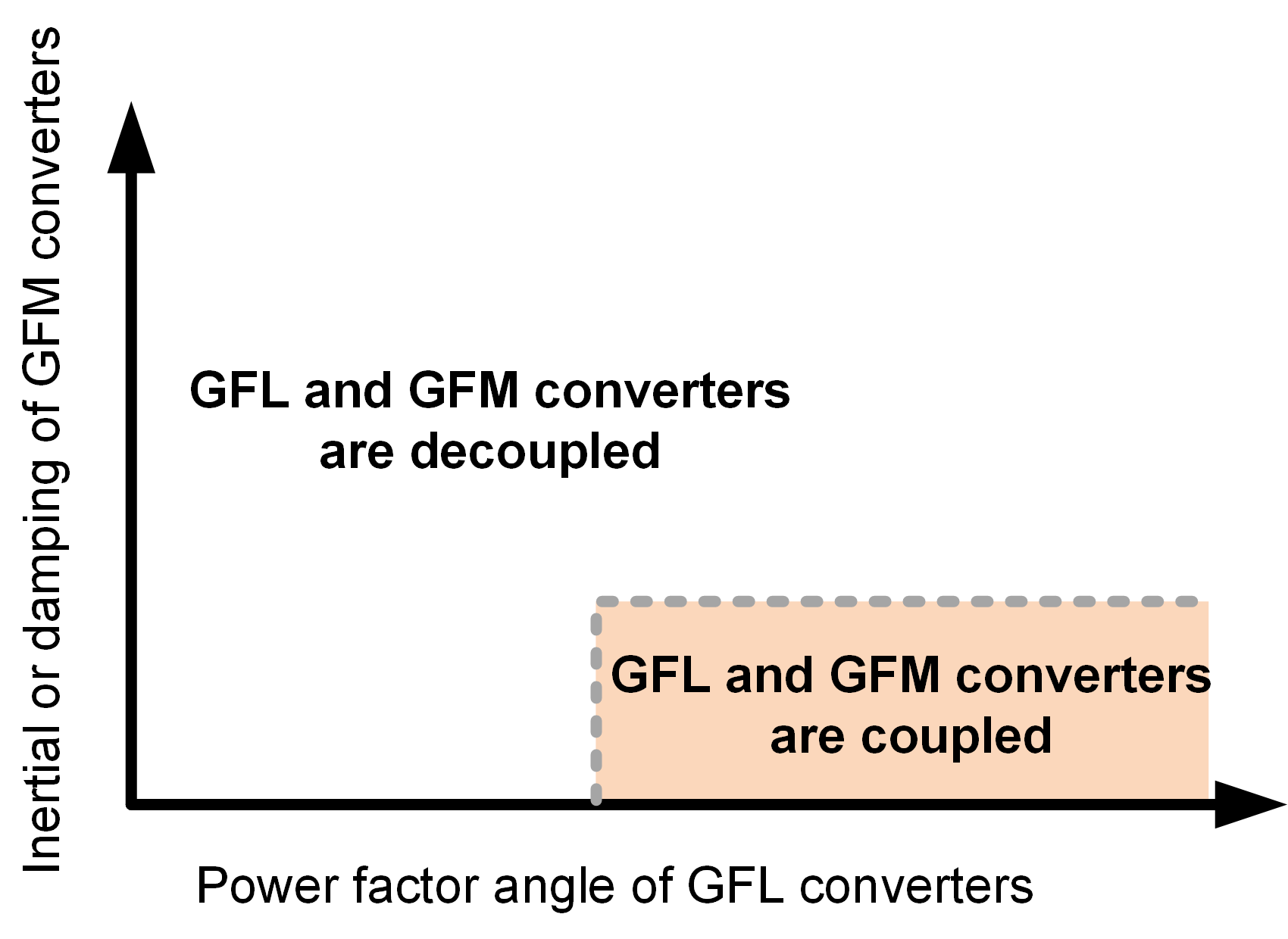}
	\vspace{-3mm}
	\caption{The decoupling conditions of GFL and GFM converters} 
	\vspace{-0.4cm}
	\label{fig7}
\end{figure}
\textcolor{catalogueblue}{Based on \eqref{eq:simple condition}, we provide an analytical lower bound for the GFM damping that ensures decoupling between GFM and GFL converters, which can be obtained through scalar calculations without frequency sweeping. This offers a practical reference for designing GFM–GFL decoupling. If a less conservative design is desired, \eqref{eq:decoupled condition} can be used, still without requiring complex calculations: it only requires calculating the eigenvalues of a static matrix, and the frequency-dependent terms involved are all one-dimensional scalar quantities. }

\textcolor{catalogueblue}{ $\delta$ represents the degree of timescale coupling between the GFM and GFL converters — a smaller $\delta$ indicates that their timescales are closer. Large inertia and damping cause the GFM’s timescale to be much slower than that of the GFL, thereby ensuring a large $\delta$ and effectively decoupling the GFM from the GFL.}

\textcolor{catalogueblue}{ $\|{\bm L}_2\|_2\|{\bm L}_3\|_2$ represents the coupling strength of the off-diagonal blocks in the synchronizing coefficient matrix ${\bm L}(s)$; a larger $\|{\bm L}_2\|_2\|{\bm L}_3\|_2$ corresponds to stronger coupling. Because the angle differences between the converters are small, i.e. $\theta_{ij}\approx 0$, if the power factor angles of GFL converters satisfy $\varphi_j \approx 0$, then $\sin(\theta_{ij}-\varphi_j)\approx 0$, ${\bm L}_{2,ij}=V_iI_jc_{ij}\sin(\theta_{ij}-\varphi_j)\approx 0$ and $\|{\bm L}_2\|_2\|{\bm L}_3\|_2\approx0$. Therefore, small power factor angle  $\varphi_j$ of GFL ensures decoupling, even if the timescales of the GFM and GFL are coupled.  }

 Thus, GFL photovoltaics and wind turbines operating at unity power factor will not couple with GFM converters in terms of synchronous dynamics. Typically, Static Var Generators (SVGs), which function as reactive power and voltage regulators, will strongly couple with low-inertia GFM converters. In the future, if GFL photovoltaics and wind turbines are given voltage regulation capabilities, they may also operate at non-unity power factor and potentially couple with low-inertia GFM converters.

To verify the validity of the decoupling conditions, we compare the decoupled system in \eqref{eq:decoupled} with the original system in \eqref{eq:2citezheng} under different parameters, based on the system shown in Fig. \ref{fig5}.
\begin{example}
Based on Example 1, we set all GFM converters have fixed $J/D=10$. The bode plots of $\alpha(s)$ with different $R/X$ ratio are shown in Fig. \ref{fig8}. \textcolor{catalogueblue}{ When within $10$Hz and $R/X$ is small, we assume $\alpha(s)\approx1$, then the elements numerical of ${\bm L}(s)$ with different GFL converters' power factor angles $\varphi$ are shown in Fig. \ref{fig9}. If the frequency exceeds 10 Hz, substitute the value of $\alpha(s)$ and calculate the magnitude of the elements in ${\bm L}(s)$.} Drawing from the physics principles of SGs, the values of ${\bm L}(s)$ reflect the converters' ability to synchronize with each other. 

 Fig. \ref{fig9} (a) with zero power factor angles demonstrates: 
\begin{enumerate}
    \item \textcolor{catalogueblue}{The off-diagonal elements of ${\bm L}_1$ are nearly zero: GFL converters with $\varphi=0$ lack the ability to synchronize with each other.}
    \item \textcolor{catalogueblue}{The elements of ${\bm L}_2$ exhibit relatively large magnitudes: GFL converters with $\varphi=0$ only synchronize by following GFM converters.}
    \item \textcolor{catalogueblue}{The elements of ${\bm L}_3$ are nearly zero: GFL converters with $\varphi=0$ do not affect the synchronization of GFM.}
    \item \textcolor{catalogueblue}{The elements of ${\bm L}_4$ exhibit large magnitudes: GFM converters can synchronize with each other internally.}
\end{enumerate}
These conclusions explain why PLL is a following control and PSC is a forming control from the synchronization perspective; GFM converters can provide equivalent ``synchronous torque", whereas GFL converters with $\varphi=0$ cannot.

But when power factor angles $\varphi=-0.5{\rm rad}$ for GFL converters, Fig. \ref{fig9} (b) shows that the values of ${\bm L}_2$ and ${\bm L}_3$ are similar in magnitude. This demonstrates that, from the synchronization perspective, GFL converters are no longer purely passive followers; they generate equivalent ``synchronous torque", both internally and in relation to GFM converters. Fig. \ref{fig9} illustrates that the power factor angles of GFL converters are key parameters influencing the coupling between GFL and GFM converters by altering the values of ${\bm L}_3$.

\begin{figure}
	\centering
	\includegraphics[width=2.5in]{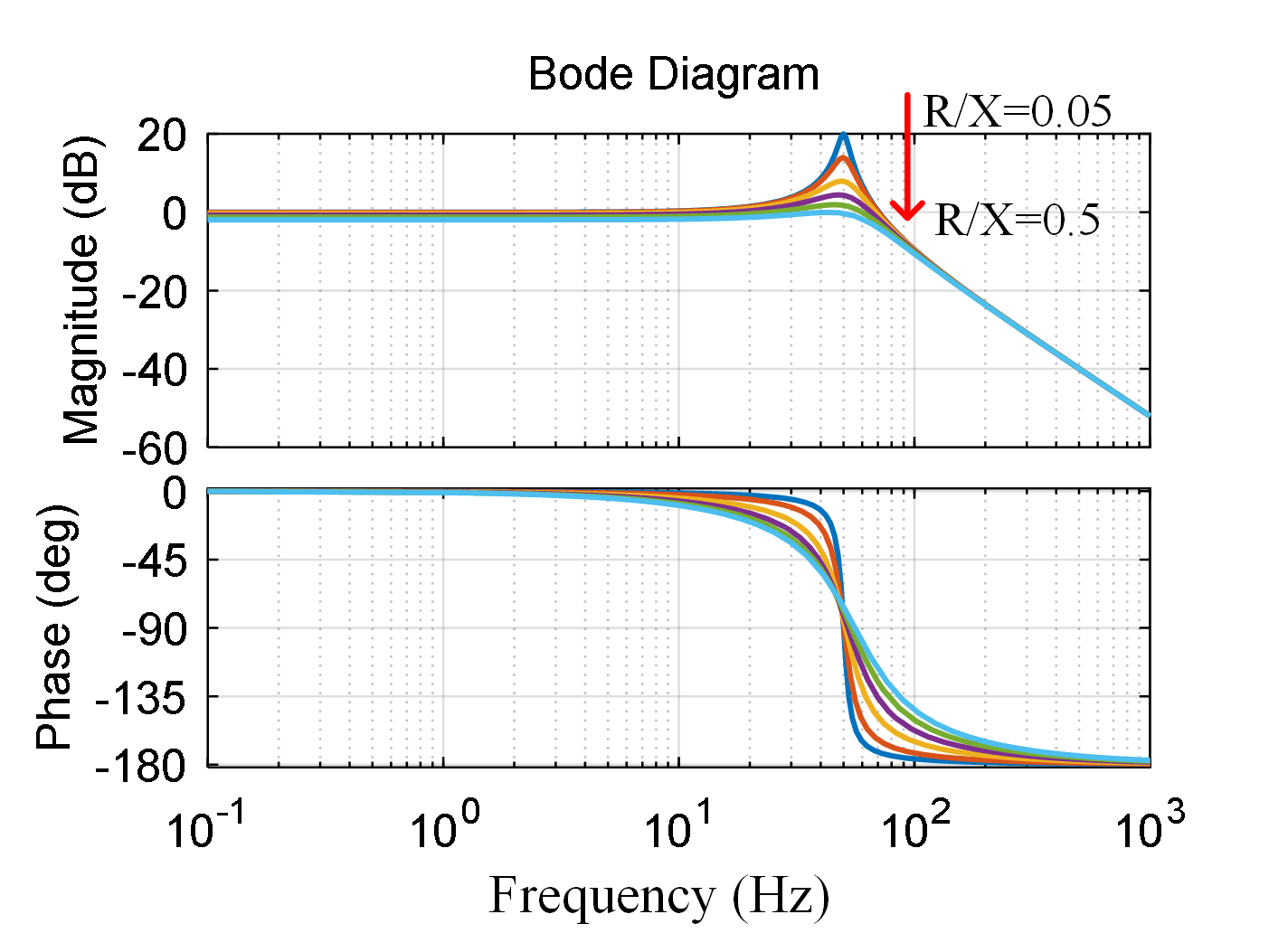}
	\vspace{-3mm}
	\caption{The bode plots of $\alpha(s)$ with different $R/X$ ratio} 
	\vspace{-0.4cm}
	\label{fig8}
\end{figure}

Furthermore, the comparison of eigenvalues between \eqref{eq:decoupled} and \eqref{eq:closedloop}, considering varying GFL converters' power factor angles and GFM converters'  inertial $J$, are shown in Fig. \ref{fig10}. When $J$ is large ($=10$p.u.) or $\varphi$ is small ($=0$), the eigenvalues of the decoupled system closely align with those of the original system, as shown in Fig. \ref{fig10} (a), (c), and (d). The decoupling conditions break down only when $J$ is small ($=1$p.u.) and $|\varphi|$ is large ($\varphi=-0.5{\rm rad}$)  simultaneously, resulting in a significant eigenvalues deviation between the decoupled and original systems, as shown in Fig. \ref{fig10} (b). 

\textcolor{catalogueblue}{In order to verify the validity of the decoupling conditions across a broader operating range, Fig. \ref{fig11} presents 3D plots of the maximum eigenvalue error and $\frac{\|{\bm L}_2\|_2\|{\bm L}_3\|_2}{\delta}$ ($\omega=58$rad, which is PLL bandwidth and oscillation frequency) under varying inertia of GFM converters and power factor angles of GFL converters. It is evident that both $\frac{\|{\bm L}_2\|_2\|{\bm L}_3\|_2}{\delta}$ and the eigenvalues error are relatively large only when both $|\varphi|$ and $J$ are sufficiently large. Otherwise, the eigenvalue error remains very small, confirming the validity of decoupling conditions.}
\begin{figure}
	\centering
	\includegraphics[width=3in]{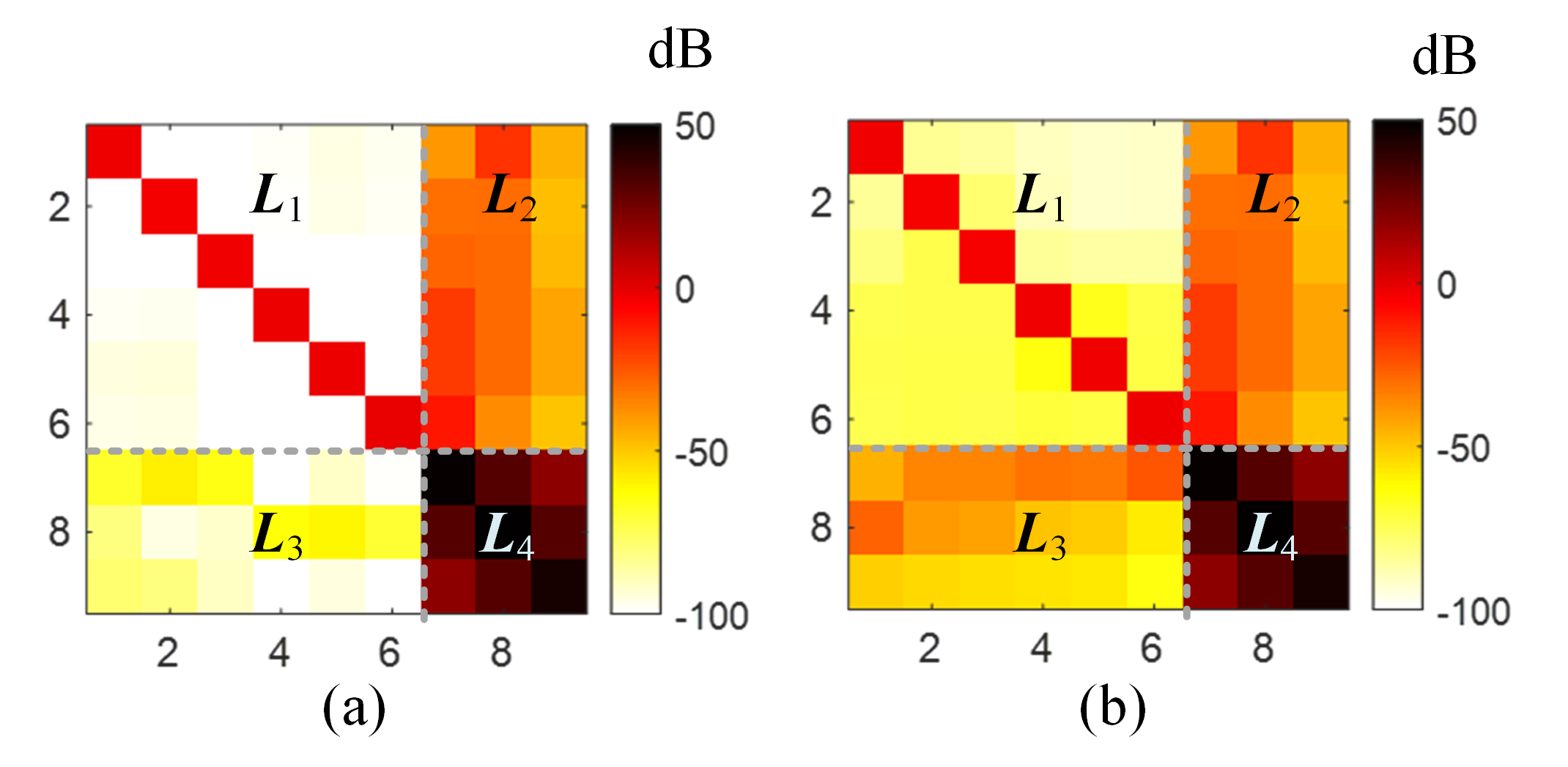}
	\vspace{-3mm}
	\caption{The elements of ${\bm L}(s)$ within 10Hz. (a) $\varphi=0$. (b) $\varphi=-0.5{\rm rad}$} 
	\vspace{-0.4cm}
	\label{fig9}
\end{figure}
\begin{figure}
	\centering
	\includegraphics[width=2.8in]{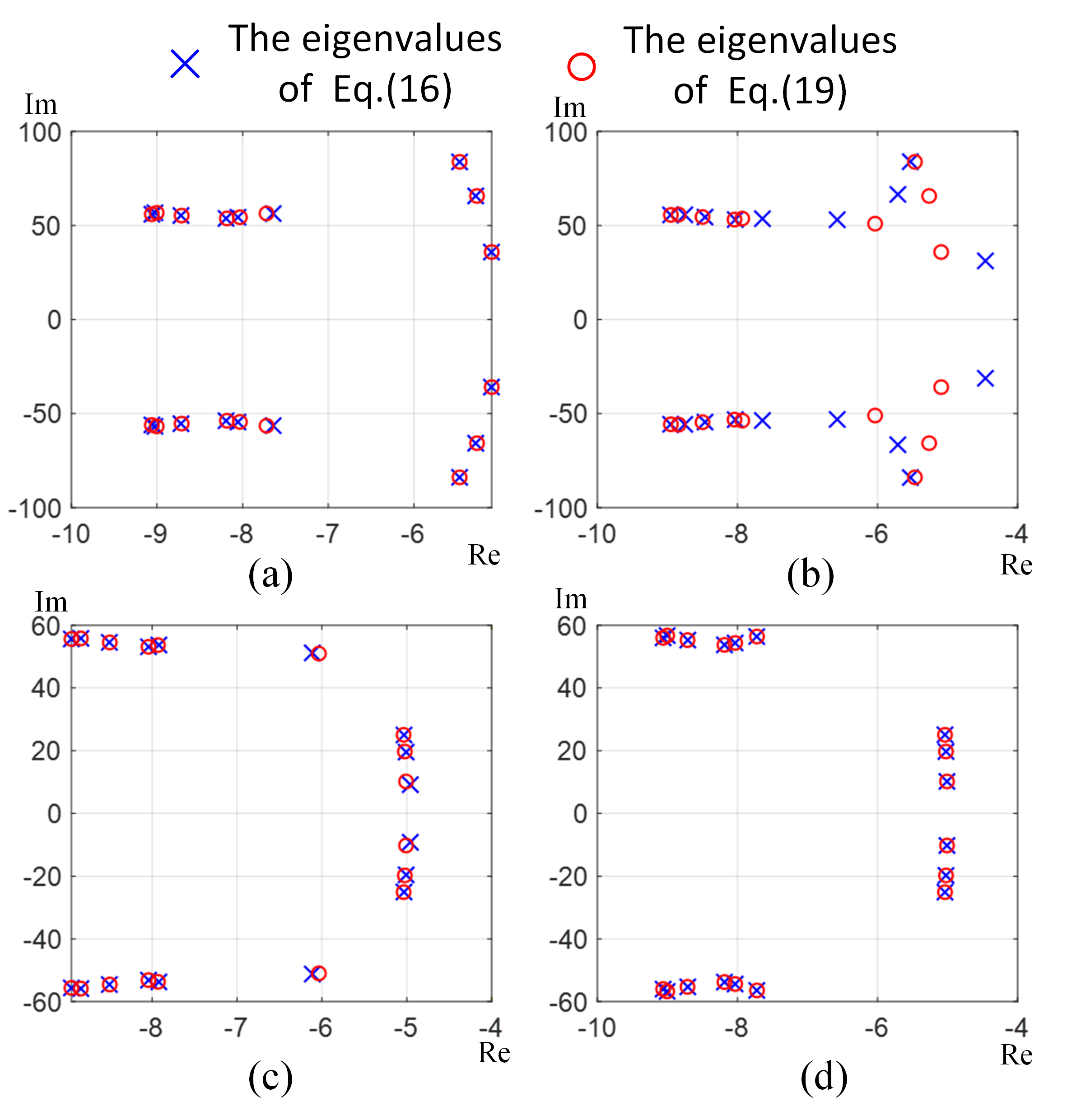}
	\vspace{-3mm}
	\caption{The eigenvalues of decoupled system and original system. (a) $\varphi=0, J=1$p.u.. (b) $\varphi=-0.5{\rm rad}, J=1$p.u.. (c) $\varphi=-0.5{\rm rad}, J=10$p.u.. (d) $\varphi=0, J=10$p.u..   } 
	\vspace{-0.4cm}
	\label{fig10}
\end{figure}
\begin{figure}
	\centering
	\includegraphics[width=3in]{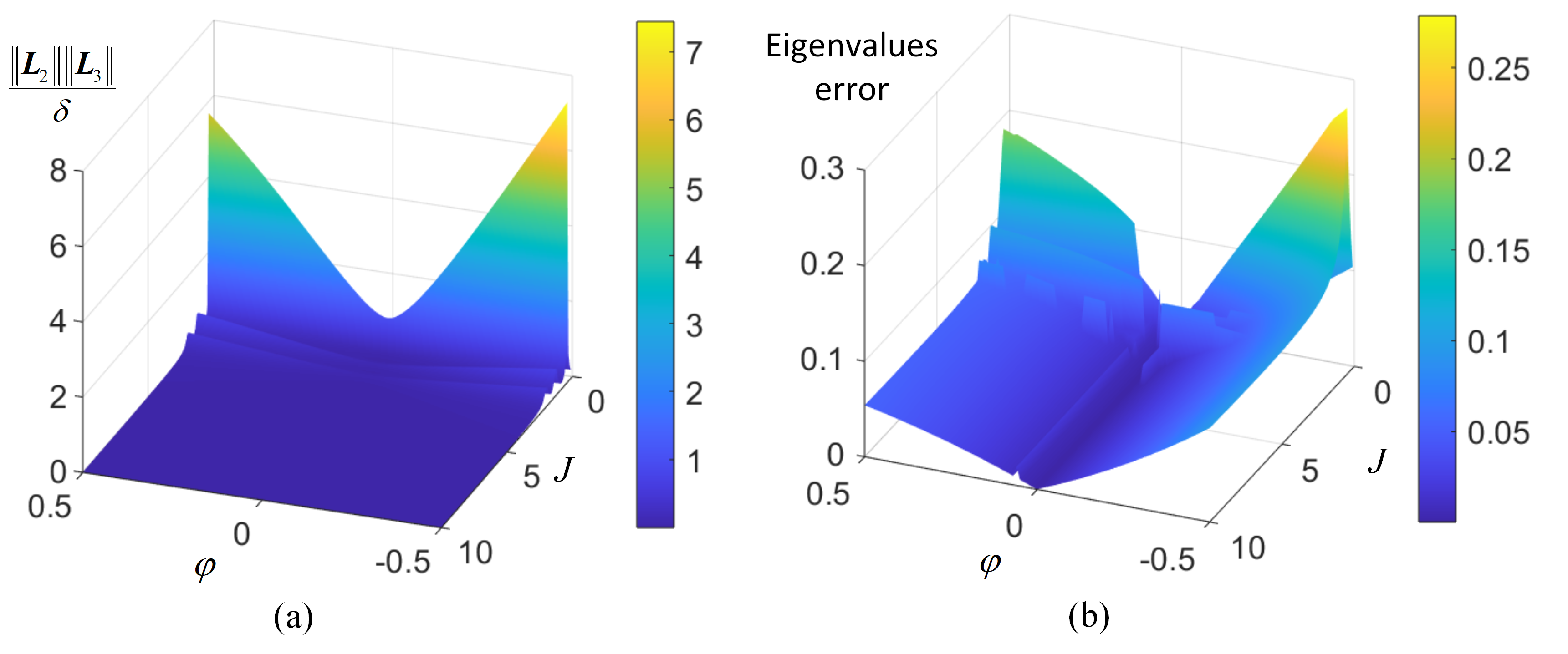}
	\vspace{-3mm}
	\caption{\textcolor{catalogueblue}{3D plots of the maximum eigenvalue error and $\frac{\|{\bm L}_2\|_2\|{\bm L}_3\|_2}{\delta}$. (a) $\frac{\|{\bm L}_2\|_2\|{\bm L}_3\|_2}{\delta}$. (b) the maximum eigenvalue error }} 
	\vspace{-0.4cm}
	\label{fig11}
\end{figure}
\end{example}

\subsection{The small-signal synchronous stability of decoupled systems}
For the decoupled systems of Eq. \eqref{eq:decoupled}, the small-signal synchronous stability mechanism can be understood as the ability of ${\bm D}_{\rm gfl}$ and ${\bm D}_{\rm gfm}(s)$ to provide positive damping.

\begin{proposition}[{\bf The small-signal stability conditions of decoupled systems}]\label{stability condition}
\textcolor{catalogueblue}{When \eqref{eq:decoupled condition} holds, the system shown in \eqref{eq:decoupled} is stable if
\begin{equation}\label{eq:static stability}
\begin{aligned}
\frac{k_{P,i}v_{d,i}}{k_{I,i}}&>-\underline{\lambda}({\bm K}^{\rm d}_{\rm gfl}),\forall i\,,\\
D_i&>-\Im[\alpha^*(j\omega)]\underline{\lambda}\left(\frac{{\bm K}^{\rm sy}_{\rm gfm}}{\omega}+j{\bm K}^{\rm d}_{\rm gfm}\right),\forall i
\,,
\end{aligned}
\end{equation}
where, $\underline{\lambda}(\cdot)$ represents minimum eigenvalue, $\forall  \omega \in [\underline{\omega}_M,\overline{\omega} _M]$ corresponds to the imaginary parts range of all solutions of ${\rm det}[{\bm N}_{\rm gfm}(s)]=0$, $\alpha^*(j\omega)=\frac{\alpha(j\omega)-\alpha(-j\omega)}{2}$, $\mathfrak{I}(\cdot)$ represents the imaginary part.}
\end{proposition}
\begin{proof}
\textcolor{catalogueblue}{See Appendix~\ref{sec:appendixC}.}
\end{proof}

\textcolor{catalogueblue}{Under the decoupled condition between the GFM and GFL converters, \eqref{eq:static stability} directly provides the converter parameter conditions required to ensure small-signal stability of the system: the converter’s positive damping is greater than the network’s negative damping. In contrast to GNC-based methods, which intertwine network and converter parameters and therefore do not yield analytical results, our proposed analytical lower bounds of the controller parameters can be obtained directly by calculating constant matrices related only to the network power flow information. Furthermore, we define $d_{\rm gfl}=\min\left(\frac{k_{P,i}v_{d,i}}{k_{I,i}}\right)+\underline{\lambda}({\bm K}^{\rm d}_{\rm gfl})$ and $d_{\rm gfm}=\min\left(D_i\right)+\Im[\alpha^*(j\omega)]\underline{\lambda}\left(\frac{{\bm K}^{\rm sy}_{\rm gfm}}{\omega}+j{\bm K}^{\rm d}_{\rm gfm}\right)$ as the margin indicators for the GFL and GFM subsystems, respectively. }

The bode plot of $\alpha(s)$ with different $\tau$ as shown in Fig.\ref{fig8}. It shows that the phase lag increases with $\tau$ and $\omega$. Therefore, if the GFM subsystem exhibits low oscillation frequency $\omega$ or network line has small $\tau$ such that $\mathfrak{I}[\alpha^*(j\omega)] \approx 0$, the GFM subsystem will be stable provided that $D_i>0, \forall i \in[n+1,n+m]$.

\begin{remark} ({\bf gSCR and network negative damping}) \label{gSCR} \textbf{gSCR} is defined as the minimum eigenvalue of network admittance matrix \cite{Huanhai:gSCR}. When the phase difference between the GFL converters are neglected and all GFL converters power factor angles are assumed to be $0$, $- \omega_0{\bm K}^{\rm d}_{\rm gfl} $ becomes a network impedance matrix. Its maximum eigenvalue is $\overline{\lambda}(- \omega_0{\bm K}^{\rm d}_{\rm gfl})=-\omega_0\underline{\lambda}({\bm K}^{\rm d}_{\rm gfl})=\frac{1}{\textbf{  gSCR} }$. Therefore, \textbf{gSCR} can qualify the negative damping of the network and $\frac{k_{P,i}v_{d,i}}{k_{I,i}} $ represents the positive damping of the GFL converter. When GFL converter positive damping is lager than network negative damping, i.e. the critical SCR of GFL converter is less than the  \textbf{gSCR}, GFL subsystem is stable. Therefore, the criterion based on $d_{\rm gfl}$ aligns with the criterion based on gSCR under this situation. 
\end{remark}

\begin{example}
Based on the system shown in Fig.\ref{fig5}, we set the constant $J=6p.u.$ and $\varphi=0$ to meet the proposed decoupled conditions. Then we adjust $k_{P,i}$ of the GFL converters and $D_{i}$ of the GFM converters to validate the stability conditions. The curves of $d_{\rm gfl},d_{\rm gfm}$ and the smallest damping ratio of the eigenvalues of two subsystems with different $k_{P,i},D_i$ are shown in Fig.\ref{fig12}. The results indicate that the proposed indexes align closely with the eigenvalue damping ratio, confirming accuracy of the proposed indexes in assessing the stability of the decoupled system.
\end{example}
\begin{figure}
	\centering
	\includegraphics[width=3in]{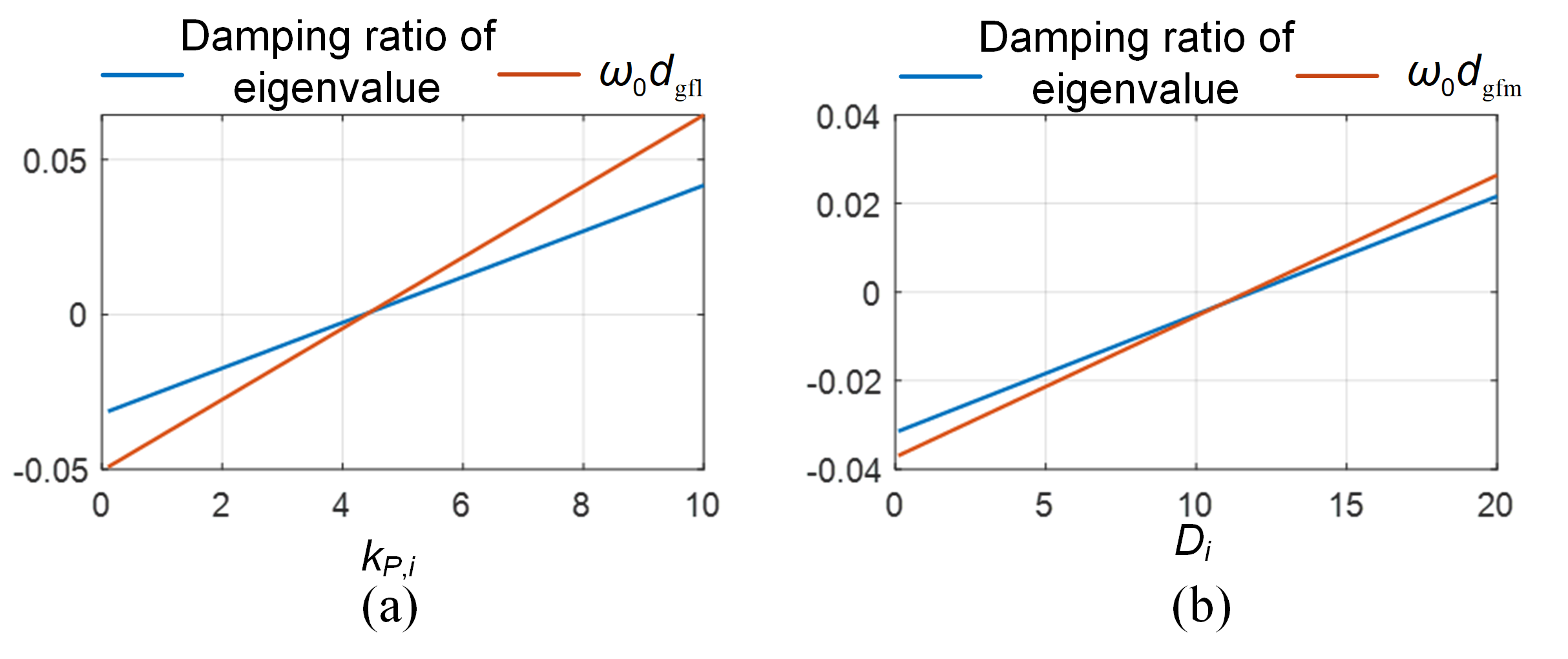}
	\vspace{-3mm}
	\caption{The proposed index and damping ratio of eigenvalue with different parameters. (a) GFL converters subsystem. (b) GFM converters subsystem.} 
	\vspace{-0.4cm}
	\label{fig12}
\end{figure}
\section{The Decentralized Stability Conditions of the Coupled System Based on Small Phase Theorem}
When GFM and GFL converters are coupled, the matrices of \eqref{eq:GNC function} are no longer Hermitian, making it difficult to apply the Proposition \ref{stability condition}. Therefore, this section introduces a unified approach to evaluate small-signal synchronous stability, applicable to both coupled and decoupled systems. 

\textcolor{catalogueblue}{Although the method in this section avoids computing the dynamic matrices of the GFL converters compared with previous methods, it should be noted that the approach in this section is more complex than \eqref{eq:static stability}, as it embeds the GFM parameters into the network while separating only the GFL control parameters. Therefore, its practical application requires more complicated computations or network identification based on data measured at the GFL converter terminals. Hence, under the condition that decoupling is achievable, we still recommend using \eqref{eq:static stability} to assess stability, as its concise form is more convenient for engineering practice.}
 
\subsection{The definition of complex matrix phase}
The phase of a complex matrix $A\in \mathbb{C}^{n\times n}$ is defined in \cite{wang2020phases} based on the {\em numerical range} $W(A)$ of a matrix,
\begin{equation}\label{eq:numerical range}
\begin{aligned}
W(A)={x^HAx: x\in \mathbb{C}^n, \|x\|=1}
\,,
\end{aligned}
\end{equation}
which is the set of all values of $x^HAx$ for any unit-norm complex vector $x$. $W(A)$ is a convex subset of $\mathbb{C}$ and contains the eigenvalues of $A$. $A$ is called a {\em sectorial matrix} if $0 \notin W(A)$ \cite{wang2020phases}.
\begin{definition}[{\bf Phase of a matrix}]\label{def:phase}
If $A$ is a sectorial matrix, it can be decomposed to $A=V^HEV$. Where, $V$ is a non-singular matrix and $E={\rm diag}(e_1,\cdots,e_n)$ is a diagonal unitary matrix. All elements of $E$ lie on an arc of the unit circle with length smaller than $\pi$. The {\em phase} range of $A$ can be defined as $\phi(A)\in[\underline{\phi}(A),\overline{\phi}(A)]$, 
\begin{equation}\label{eq:small phase}
\begin{aligned}
\overline{\phi}(A)=&{\rm max} [\phi(e_1),\cdots,\phi(e_n)],\\
\underline{\phi}(A)=&{\rm min} [\phi(e_1),\cdots,\phi(e_n)]
\,,
\end{aligned}
\end{equation}
with $\overline{\phi}(A)-\underline{\phi}(A)<180^\circ$. $\phi(\cdot)$ represents taking the phase of the complex number.
\end{definition}
\begin{figure}
	\centering
	\includegraphics[width=3in]{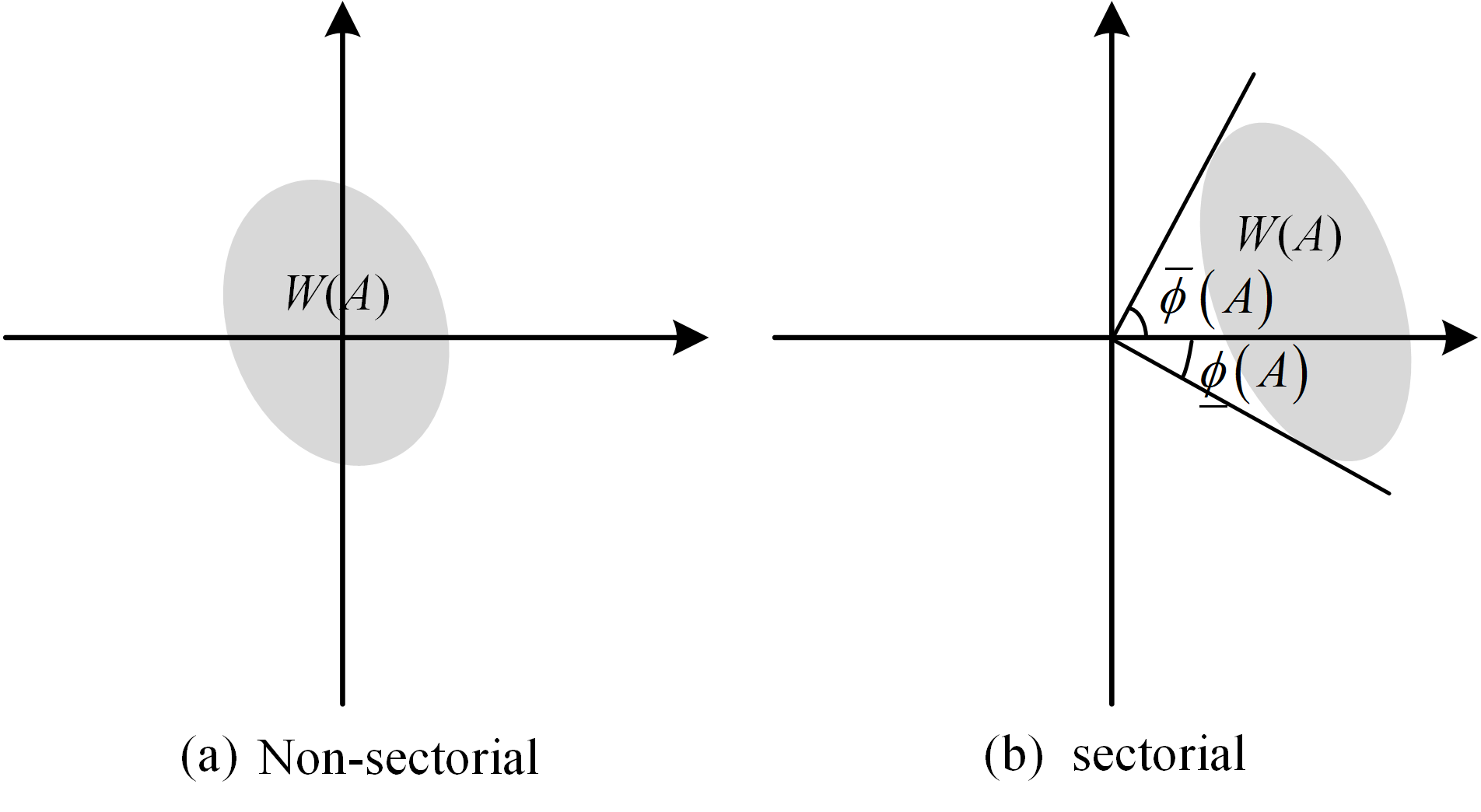}
	\vspace{-3mm}
	\caption{A non-sectorial matrix and phase of a sectorial matrix. (a) Non-sectorial. (b) sectorial.} 
	\vspace{-0.4cm}
	\label{fig13}
\end{figure}
The graphic illustrations of the numerical range $W(A)$ and the phases of $A$ are shown in Fig. \ref{fig13}. The two angles between the tangent from the origin to $W(A)$ and the horizontal axis in Fig. \ref{fig13} (b) are $\overline{\phi}(A)$ and $\underline{\phi}(A)$, respectively. All elements of $E$ lie in $W(A)$. If $A$ is a sectorial matrix, we can obtain 
\begin{equation}\label{eq:phase calculation}
\begin{aligned}
A(A^H)^{-1}=V^HE^2(V^H)^{-1}
\,.
\end{aligned}
\end{equation}

Therefore, the phase of $A$ (i.e. the phase of the elements of $E$) are the halves of the phases of the eigenvalues of $A(A^H)^{-1}$. The eigenvector matrix of $A(A^H)^{-1}$ is $V^H$. More details on phase calculation can be found in \cite{wang2020phases}.

\subsection{The small phase theorem}
We denote the characteristic polynomial of a feedback system by ${\rm det}(I+GH)=0$ with $G,H \in \mathcal{RH}^{m \times m}_\infty $ , whose block diagram is shown in Fig. \ref{fig14}. The following theorem provides a sufficient stability condition for the feedback system in Fig. \ref{fig14}.
\begin{figure}
	\centering
	\includegraphics[width=2in]{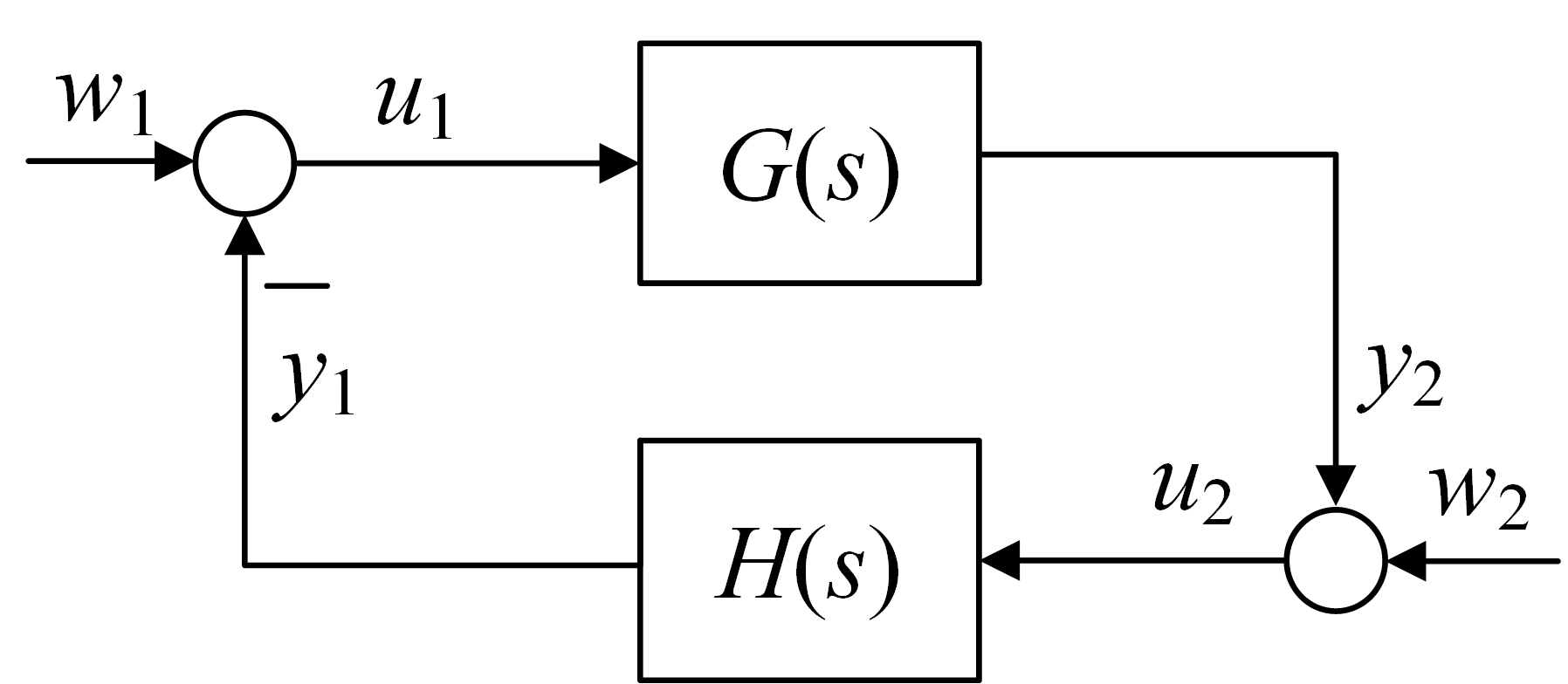}
	\vspace{-3mm}
	\caption{A feedback system.} 
	\vspace{-0.4cm}
	\label{fig14}
\end{figure}
\begin{lemma} \label{small phase}
Let the open-loop systems $G(s)$ and $H(s)$ are real, rational, stable, and proper transfer function matrices, if for each $\omega \in [0,\infty)$, satisfy $G(j\omega)$ and $H(j\omega)$ are sectorial and
\begin{equation}\label{eq:small phase theroy}
\begin{aligned}
\overline{\phi}[G(j\omega)]+\overline{\phi}[H(j\omega)]<&180^\circ,\, {\rm and}\\
\underline{\phi}[G(j\omega)]+\underline{\phi}[H(j\omega)]>&-180^\circ
\,,
\end{aligned}
\end{equation}
the feedback system in Fig. \ref{fig14} is stable \cite{chen2019phase}.
\end{lemma}
\subsection{The decentralized stability condition of coupled system}
\textcolor{catalogueblue}{When the decoupling condition is not satisfied, the system \eqref{eq:closedloop} can be divided into two coupled subsystems through the Schur complement,
\begin{equation}\label{eq:Shcur}
\begin{aligned}
\left\{\begin{matrix}
{\rm det}\left[I_n+\tilde{{\bm G}}_{\rm gfl}(s)\left({\bm L}_1+{\bm K}^d_{\rm gfl}s-{\bm L}_2{\bm N}^{-1}_{\rm gfm}(s){\bm L}_3\right)\right]=0\\
{\rm det}[{\bm N}_{\rm gfm}(s)]=0
\end{matrix}\right.
,
\end{aligned}
\end{equation} 
where $\tilde{\bm G}_{\rm gfl}(s)={\rm diag}\left\{\frac{T_{P,i}s+1}{T_{J,i}s^2}\right\},i \in [1,n]$.}

\textcolor{catalogueblue}{For GFM subsystem, the characteristic polynomial is still ${\rm det}[{\bm N}_{\rm gfm}(s)]=0$, we can still use \eqref{eq:static stability} to qualify the stability of the GFM subsystem. In addition, for the stability analysis of the coupled GFL subsystem, we propose the following model and proposition.}

\textcolor{catalogueblue}{We consider the coupled feedback system shown in \eqref{eq:Shcur}, the GFL subsystem characteristic polynomial is rewritten as (the derivation can be found in Appendix \ref{sec:appendixD}.)
\begin{equation}\label{eq:small phase theroy function}
\begin{aligned}
{\rm det}\left[{\bm I}_{n}+{\bm G}_{\rm gfl}(s)\left({\bm L}_1-\widetilde{{\bm L}}(s)-0.1{\bm V}^{\rm gfl}_d\right)\right]=0
\,,
\end{aligned}
\end{equation}
where, $\widetilde{{\bm L}}(s)=-{\bm K}^d_{\rm gfl}s+{\bm L}_2{\bm N}^{-1}_{\rm gfm}(s){\bm L}_3$.}

The original transfer function of a GFL converter is $\frac{T_{P,i}s+1}{T_{J,i}s^2}$. Based on this, we incorporate $0.1{\bm V}^{\rm gfl}_d$ into GFL converters dynamics, resulting in a new equivalent transfer function ${\bm G}_{\rm gfl}(s)={\rm diag}\left\{\frac{T_{P,i}s+1}{T_{J,i}s^2+0.1T_{P,i}{\bm V}^{\rm gfl}_{d,i}s+0.1{\bm V}^{\rm gfl}_{d,i}}\right\},i \in [1,n]$. This modification helps to avoid the converter's transfer function phase approaching –180° in the low-frequency range, thus maintaining the validity of phase calculation and addressing the phase degradation problem noted in \cite{Linbin:gain_phase} . 

\textcolor{catalogueblue}{Based on Lemma~\ref{small phase}, let $G(s)={\bm G}_{\rm gfl}(s)$, $H(s)={\bm L}_1-\widetilde{{\bm L}}(s)-0.1{\bm V}^{\rm gfl}_d$. Because ${\bm G}_{\rm gfl}(s)$ is a diagonal matrix, $\underline{\phi}[{\bm G}_{\rm gfl}(j\omega)]={{\rm min}}\left\{\phi[{\bm G}_{{\rm gfl},i}(j\omega)]\right\}$ and  $\overline{\phi}[{\bm G}_{\rm gfl}(j\omega)]={{\rm max}}\left\{\phi[{\bm G}_{{\rm gfl},i}(j\omega)]\right\}$. The phase of ${\bm G}_{\rm gfl}(s)$ is given by $  \phi[{\bm G}_{{\rm gfl},i}(j\omega)]=$
\begin{equation}\label{eq:arc}
  \tan^{-1}(T_{P,i}\omega+1)-\tan^{-1}\left(\frac{T_{P,i}\omega}{0.1v_{d,i}-T_{J,i}\omega^2}\right) \,,
\end{equation}
with $T_{P}=k_P/k_I$, $T_J=1/k_I$. Then we can obtain the following proposition.}
\begin{proposition}[{\bf The small-signal stability conditions of coupled systems}]
According to Lemma \ref{small phase}, for each $\omega \in [0,\infty)$, if
\begin{equation}\label{eq:stability}
\begin{aligned}
\underset{i\in[1,n+m]}{{\rm max}}\left\{\phi[{\bm G}_{\rm eq,i}(j\omega)]\right\}<& \underbrace{180^\circ-\overline{\phi}[{\bm L}_1-\widetilde{{\bm L}}(j\omega)-0.1{\bm V}^{\rm gfl}_d]}_{=:\overline{\phi}_{\rm net}},\,\\
\underset{i\in[1,n+m]}{{\rm min}}\left\{\phi[{\bm G}_{\rm eq,i}(j\omega)]\right\}>&\underbrace{-180^\circ-\underline{\phi}[{\bm L}_1-\widetilde{{\bm L}}(j\omega)-0.1{\bm V}^{\rm gfl}_d]}_{=:\underline{\phi}_{\rm net}}
\,,
\end{aligned}
\end{equation}
are satisfied simultaneously, the system of \eqref{eq:small phase theroy function} is stable. 
\end{proposition}
\textcolor{catalogueblue}{The equivalent transformation after applying the Schur complement treats the GFL control parameters as the converter side while embedding the GFM information into the network, making ${\bm L}_1-\widetilde{{\bm L}}(s)-0.1{\bm V}^{\rm gfl}_d$ as an equivalent network model. $\left[\underline{\phi}_{\rm net},\overline{\phi}_{\rm net}\right]$ is referred to as the phase area of the equivalent network.} Therefore, when $d_{\rm gfm}>0$ and \eqref{eq:stability} are satisfied simultaneously, the system of Eq. \eqref{eq:Shcur} is stable.

\begin{example}
Based on the system in Fig.\ref{fig5}, we set GFM converters with $J=1p.u.$ and GFL converters with $\varphi=-0.5{\rm rad}$. The phase area of the equivalent network and GFL converters phases are shown in Fig.\ref{fig15}. 

We can observe that the stability of the coupled system can be inferred from the distance between the GFL converters' phases and the equivalent network phase boundary. When $k_P=2$, the system is unstable; when $k_P=3$, the system is critically stable; and when $k_P=20$, the system has a large stability margin. This is generally consistent with the conclusions drawn in Fig.\ref{fig12}, but the value of $k_P$ at the stability boundary is different . 
\begin{figure}
	\centering
	\includegraphics[width=3in]{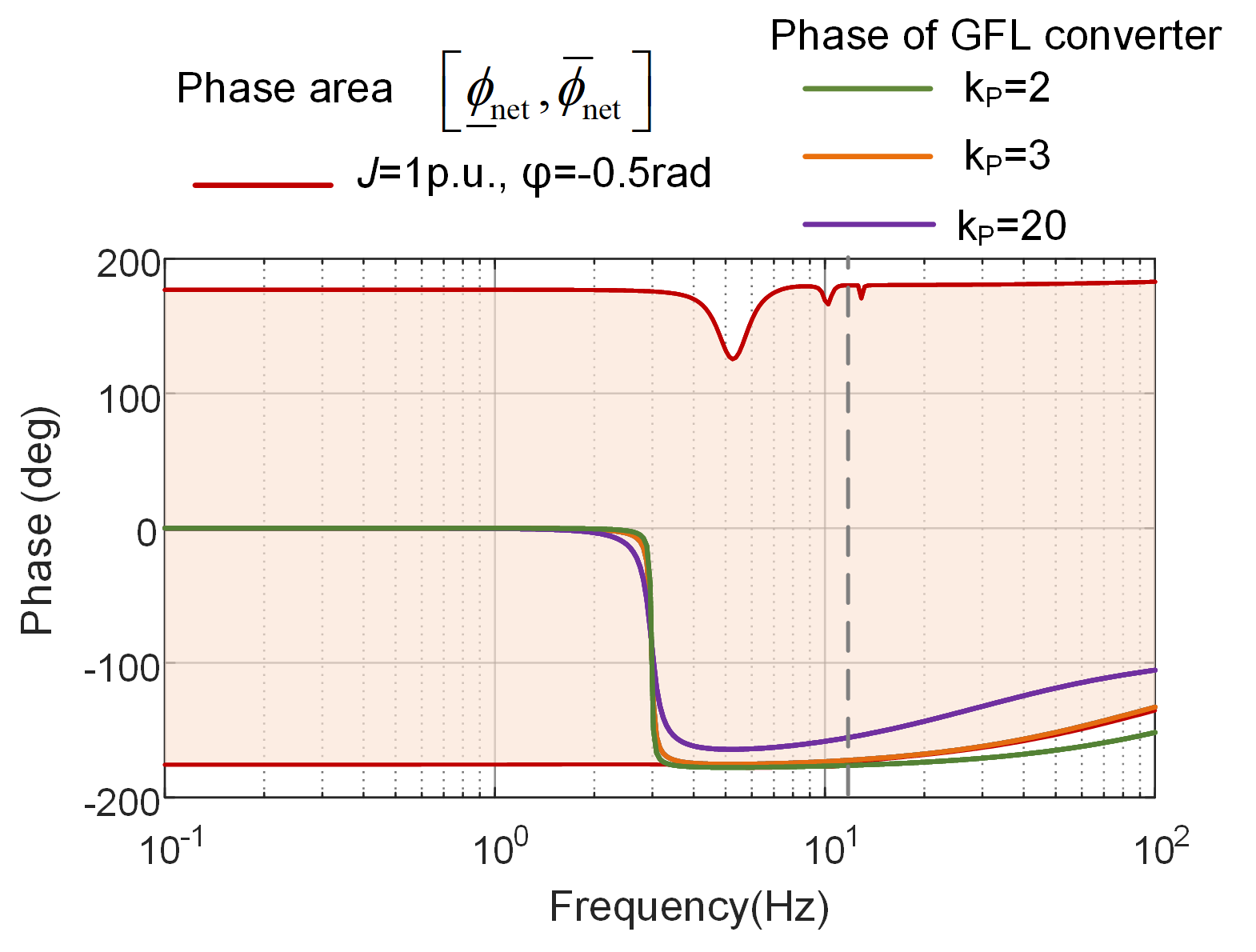}
	\vspace{-3mm}
	\caption{The phase area of equivalent network and GFL converters.} 
	\vspace{-0.4cm}
	\label{fig15}
\end{figure}
\end{example}
\section{Time-Domain Simulation}

Based on the IEEE 39 buses system shown in Fig. \ref{fig5}, we take five-case simulation in the MATLAB/Simulink to verify the the model accuracy and stability condition. The time domain responses are shown in Fig.\ref{fig16}. 

\textcolor{catalogueblue}{Case 1 (fixed $J=3$p.u. and $D=15$p.u. of GFM converters, fixed $k_P=3$ and $\varphi=0$ of GFL converters): In Case 1, we compare the time-domain responses of the proposed simplified model, the full-order detailed model, and the electromechanical model. Fig.~\ref{fig16} shows the active power response of VSC1 when the external grid voltage experiences a 0.02p.u. voltage dip. The electromechanical model exhibits small oscillations that quickly converge, while both the proposed model and the detailed model diverge and become unstable. Although there are slight differences between the two, their overall trends are consistent. This verifies that the proposed model can accurately analyze small-signal synchronous stability, consistent with the eigenvalue analysis results shown in Fig.~\ref{fig6}. In contrast, the electromechanical model fails to capture the small-signal instability of the GFL converter because the GFL oscillation frequency lies outside the electromechanical timescale. }
\begin{figure}
	\centering
	\includegraphics[width=3.5in]{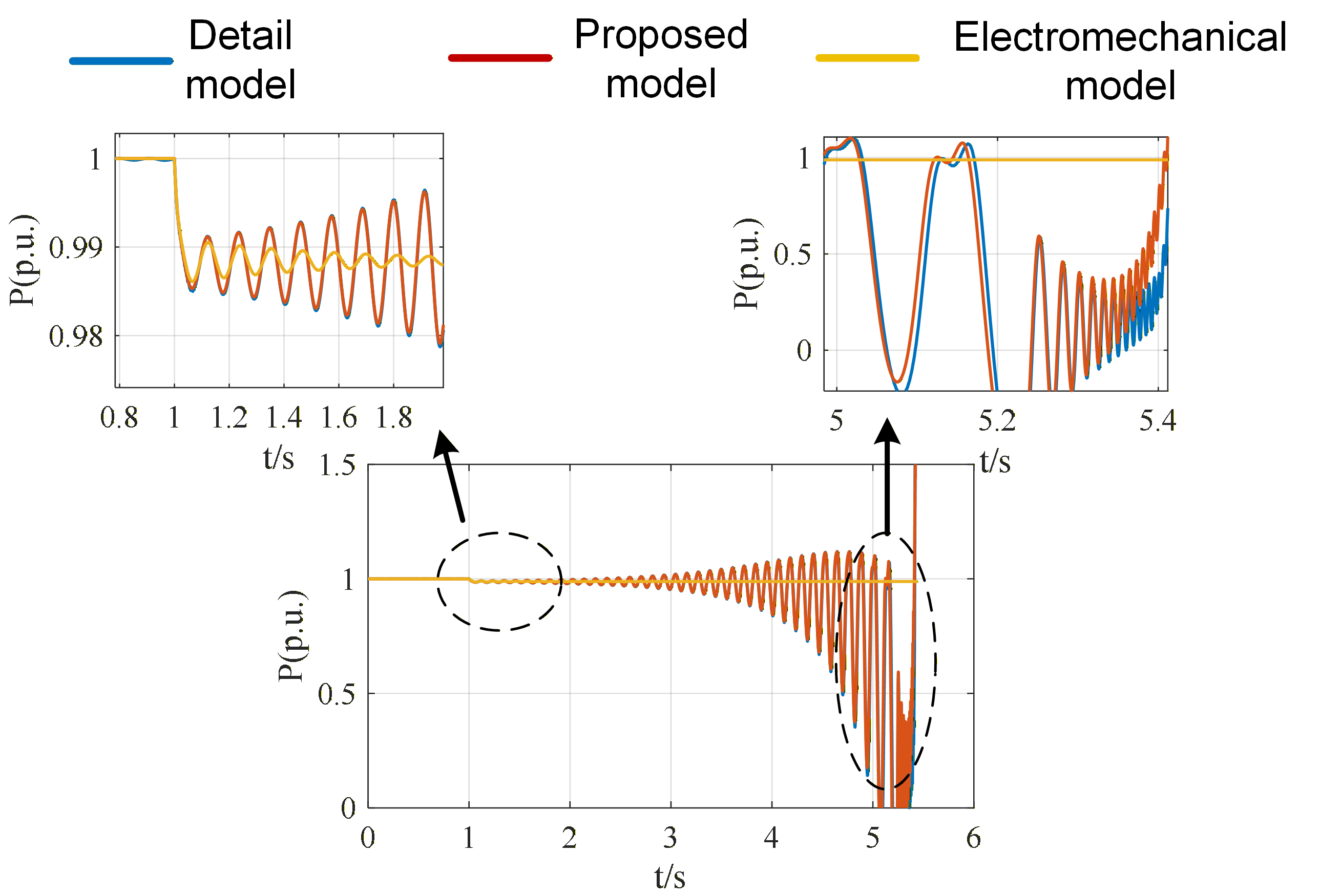}
	\vspace{-3mm}
	\caption{\textcolor{catalogueblue}{Comparison of time-domain responses among different models in case 1.}} 
	\vspace{-0.4cm}
	\label{fig16}
\end{figure}

Case 2 (fixed $J=6$p.u. and varying $D$ of GFM converters, fixed $k_P=20$ and $\varphi=0$ of GFL converters): As shown in Fig.\ref{fig17} , with $D$ decreasing, the oscillation amplitude increases. When $D=13$p.u., the system is close to critical stability. When $D=12$p.u., the system is unstable. The simulation results align with Fig.\ref{fig12} (a) of Example 3.

Case 3 (fixed $J=6$p.u. and $D=60$p.u. of GFM converters, varying $k_P$ and fixed $\varphi=0$ of GFL converters): As shown in Fig.\ref{fig18}, with $k_P$ decreasing, the oscillation amplitude increases. When $k_P=5$, the system is close to critical stability. When $k_P=4$, the system is unstable. The simulation results are aligned with Fig.\ref{fig12} (b) of Example 3.

Case 4 (fixed $J=1$p.u. and $D=10$p.u. of GFM converters, varying $k_P$ and fixed $\varphi=-0.5{\rm rad}$ of GFL converters): As shown in Fig.\ref{fig19}, with $k_P$ decreasing, the oscillation amplitude increases. When $k_P=3$, the system is close to critical stability. When $k_P=2$, the system is unstable. The simulation results align with Fig.\ref{fig15} of Example 4.

\textcolor{catalogueblue}{Case 5 (varying $J$ and varying $D$ of GFM converters, fixed $k_P=20$ and varying $\varphi$ of GFL converters): To further verify the effectiveness of the proposed stability conditions under different scenarios, simulations were conducted for two sets: $\varphi=-0.35$rad, $J=4$p.u. and $\varphi=-0.2$rad, $J=10$. As shown in Fig.~\ref{fig11}, both two satisfy the decoupling condition. The analytically calculated critical damping coefficients of the GFM converters are 11.8p.u. and 9.6p.u., respectively. The simulated critical damping coefficients are 12.4p.u. and 9.5p.u., with the corresponding time-domain waveforms shown in Fig.~\ref{fig20}~(a) and~(b). It can be observed that the waveforms exhibit equal-amplitude oscillations at the critical stability point. The errors between the analytical and simulated critical parameters are only 4.84\% and 1.04\%, respectively, further confirming the validity of the proposed stability criterion.}

Whether the GFM and GFL converters are coupled or decoupled, the proposed stability analysis method is consistent with the simulation results, which validates the effectiveness of the proposed approach. Moreover, since the stability conditions are directly based on converter parameters, they can be used to tune the synchronization control parameters of GFL and GFM converters. 
\begin{figure*} [t]
	\centering
	\includegraphics[width=6in]{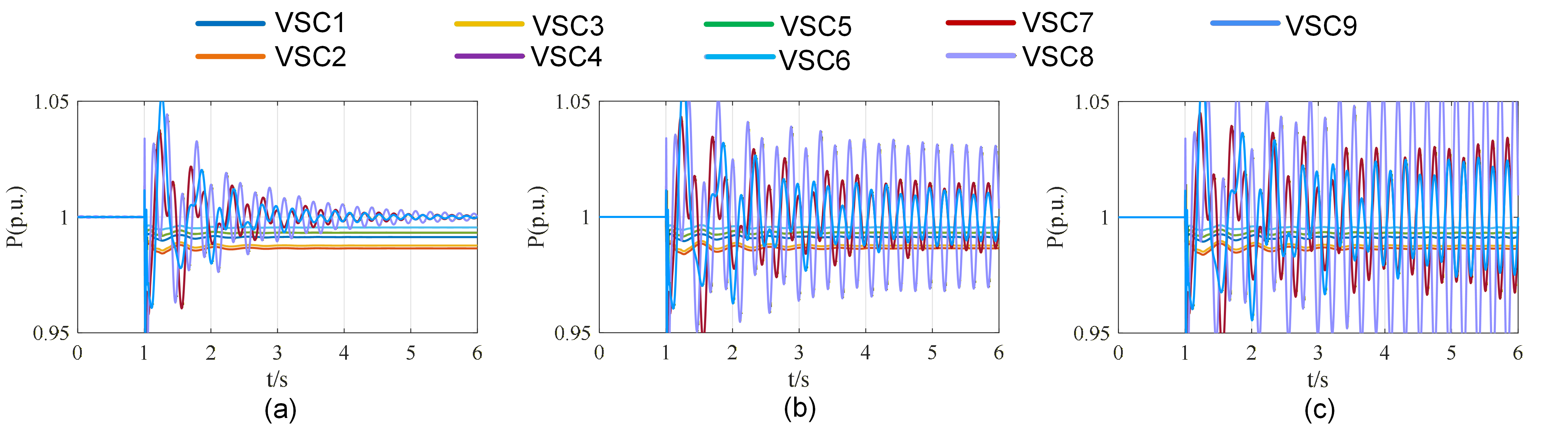}
	\vspace{-3mm}
	\caption{Time domain response of the simulation of case 2. (a) $J=6$p.u., $D=20$p.u., $k_P=20$, $\varphi=0$. (b) $J=6$p.u., $D=13$p.u., $k_P=20$, $\varphi=0$. (c) $J=6$p.u., $D=11$p.u., $k_P=20$, $\varphi=0$.} 
	\vspace{-0.4cm}
	\label{fig17}
\end{figure*}
\begin{figure*} [t]
	\centering
	\includegraphics[width=6in]{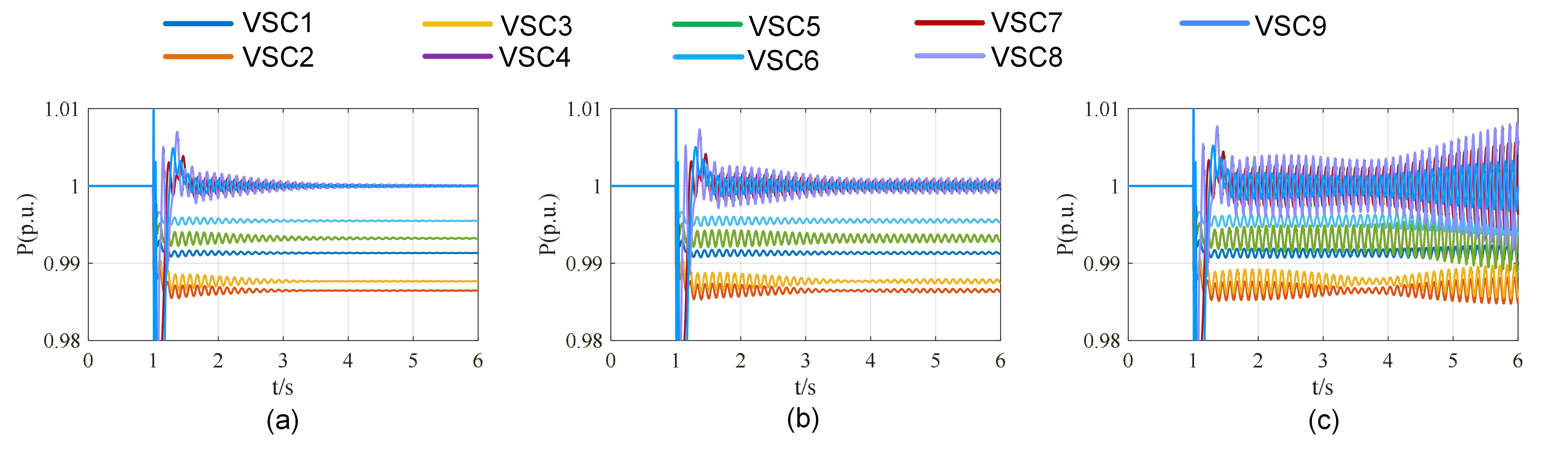}
	\vspace{-3mm}
	\caption{Time domain response of the simulation of case 3. (a) $J=6$p.u., $D=60$p.u., $k_P=6$, $\varphi=0$. (b) $J=6$p.u., $D=60$p.u., $k_P=5$, $\varphi=0$. (c) $J=6$p.u., $D=60$p.u., $k_P=4$, $\varphi=0$.} 
	\vspace{-0.4cm}
	\label{fig18}
\end{figure*}
\begin{figure*} [t]
	\centering
	\includegraphics[width=6in]{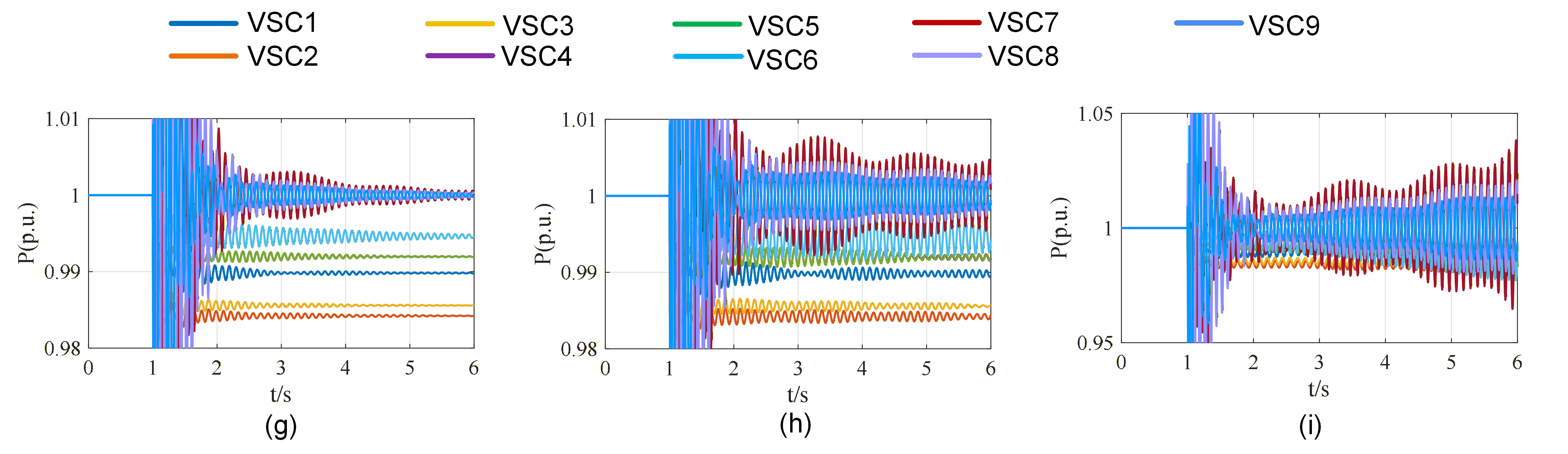}
	\vspace{-3mm}
	\caption{Time domain response of the simulation of case 4. (a) $J=1$p.u., $D=10$p.u., $k_P=4$, $\varphi=-0.5{\rm rad}$. (b) $J=1$p.u., $D=10$p.u., $k_P=3$, $\varphi=-0.5{\rm rad}$. (c) $J=1$p.u., $D=10$p.u., $k_P=2$, $\varphi=-0.5{\rm rad}$.} 
	\vspace{-0.4cm}
	\label{fig19}
\end{figure*}
\begin{figure}
	\centering
	\includegraphics[width=3in]{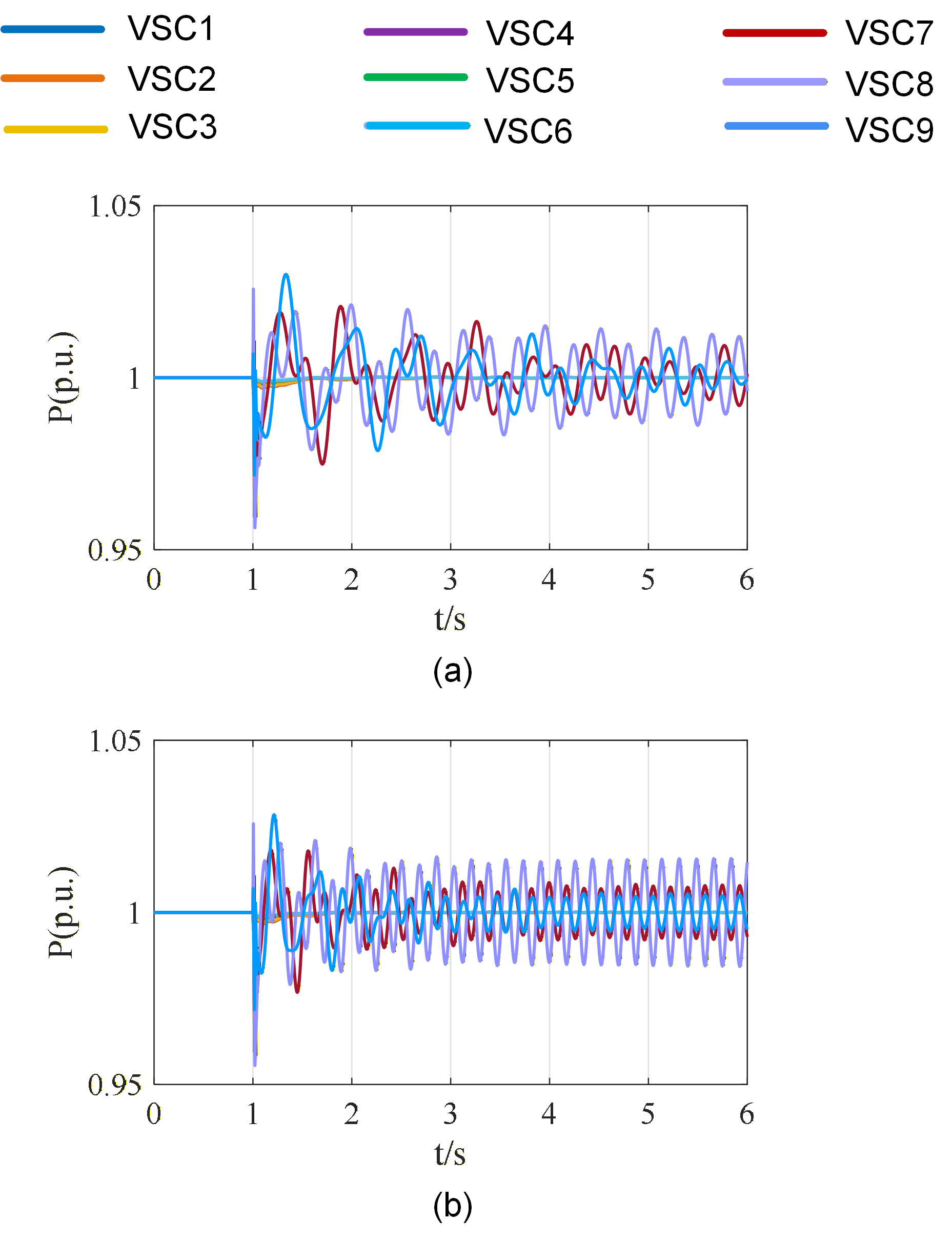}
	\vspace{-3mm}
	\caption{\textcolor{catalogueblue}{Time domain response of the simulation of case 5. (a) $J=4$p.u., $D=12.4$p.u., $k_P=20$, $\varphi=-0.35{\rm rad}$. (b) $J=10$p.u., $D=9.5$p.u., $k_P=20$, $\varphi=-0.2{\rm rad}$. }} 
	\vspace{-0.4cm}
	\label{fig20}
\end{figure}

\section{Conclusions}
\textcolor{catalogueblue}{In this paper, a generalized second-order synchronous stability analysis model is established for power systems containing both GFL and GFM converters. The coupling mechanism between GFM and GFL converters and synchronization stability conditions are analyzed, and the main conclusions are as follows.}

\textcolor{catalogueblue}{1) The proposed model is consistent with the classical second-order model of multiple SGs, allowing the synchronization stability of GFM–GFL hybrid power systems to be directly interpreted using the inertia matrix, damping matrix, and synchronizing coefficient matrix. The synchronizing coefficient matrix reflects the system’s ability to provide synchronizing torque, offering an intuitive synchronization-based interpretation of “forming” (actively providing synchronizing torque) and “following” (unable to provide synchronizing torque and passively following). Whether the damping torque provided by the damping coefficient matrix is positive determines the stability of the system. }

\textcolor{catalogueblue}{2) Based on subspace perturbation theory, the quantitative conditions for coupling between GFM and GFL converters are revealed, providing analytical lower bounds for parameter design to achieve their decoupling. The study shows that the synchronous dynamics of GFM and GFL become coupled only when both the time-scale coupling and the synchronizing coefficient matrix coupling occur simultaneously; otherwise, they can be analyzed independently. The inertia and damping of the GFM determine the coupling of time scales, while the power factor of the GFL determines the coupling of the synchronizing coefficient matrix.}

\textcolor{catalogueblue}{3) For systems where GFM and GFL converters can be decoupled, we derive parameterized decentralized stability conditions by analyzing the properties of the damping coefficient matrix. For coupled systems, we introduce the small-phase theorem to derive decentralized stability criteria as well. These conditions can be directly applied to guide parameter design, as they separate the converter control parameters from the network information. It is recommended, based on the proposed decoupling quantitative conditions, to have the GFM provide sufficient inertia and damping and then use the decoupled stability quantitative conditions to guide controller parameter design. This is because, under coupling, the network dynamics are more complex, making the decoupled conditions more convenient and practical for engineering applications. All theoretical analyzes are validated using simulations on the IEEE 39-bus system. }

In the future, more studies are needed to explore the large-signal synchronous stability characteristics of GFM and GFL converter hybrid systems.

{\appendices
\section{The derivation of network model}
\label{sec:appendixA}
According to \eqref{eq:Ygrid}, $\Delta {\bm I}^{\rm gfl}_{xy}=\left[{\bm Y}_1 \otimes \gamma(s)\right]\Delta {\bm V}^{\rm gfl}_{xy}+\left[{\bm Y}_2 \otimes \gamma(s)\right]\Delta {\bm V}^{\rm gfm}_{xy}$, we can obtain
\begin{equation}\label{eq:ytr}
    \Delta {\bm V}^{\rm gfl}_{xy}=\left[{\bm Y}_1 \otimes \gamma(s)\right]^{-1}\Delta {\bm I}^{\rm gfl}_{xy}-\left[{\bm Y}^{-1}_1{\bm Y}_2 \otimes I_2\right]\Delta {\bm V}^{\rm gfm}_{xy} 
\end{equation}

Then, substitute \eqref{eq:ytr} into $\Delta {\bm I}^{\rm gfm}_{xy}=\left[{\bm Y}_3 \otimes \gamma(s)\right]\Delta {\bm V}^{\rm gfl}_{xy}+\left[{\bm Y}_4 \otimes \gamma(s)\right]\Delta {\bm V}^{\rm gfm}_{xy}$, 
\begin{equation}\label{eq:ytr2}
\begin{aligned}
     \Delta {\bm V}^{\rm gfl}_{xy}=&\left[{\bm Y}_3{\bm Y}_1 \otimes I_2\right]\Delta {\bm I}^{\rm gfl}_{xy}-\\
     &\left[{\bm Y}_4-{\bm Y}_3{\bm Y}^{-1}_1{\bm Y}_2 \otimes \gamma(s)\right]\Delta {\bm V}^{\rm gfm}_{xy}\,.
\end{aligned}
\end{equation}

Based on \eqref{eq:ytr} and \eqref{eq:ytr2}, Eq. \eqref{eq:ZY} holds.

 \section{\textcolor{catalogueblue}{The proof of Proposition~\ref{The decoupling condition}}}
 \label{sec:appendixB}
\textcolor{catalogueblue}{According to the Davis-Kahan subspace perturbation theory in \eqref{eq:subspace}, let $A={\rm diag}({\bm N}_{\rm gfl}(s),{\bm N}_{\rm gfm}(s))$ and $E={\rm antidiag}({\bm L}_2,{\bm L}_3)$. Because both ${\bm N}_{\rm gfl}(s)$ and ${\bm N}_{\rm gfm}(s)$ are approximately symmetric, the condition number of $A$ is $\kappa(A)\approx1$, thus if $\frac{\|{\bm L}_2\|_2\|{\bm L}_3\|_2}{\min|{\lambda}({\bm N}_{\rm gfl}(s))-{\lambda}({\bm N}_{\rm gfm}(s))|}<\epsilon$, then the system \eqref{eq:2citezheng} can be decoupled to subsystem \eqref{eq:decoupled}.}

\textcolor{catalogueblue}{Because both ${\bm N}_{\rm gfl}(s)$ and ${\bm N}_{\rm gfm}(s)$ are approximately symmetric, we obtain
\begin{equation}\label{eq:jieoupro}
\begin{aligned}
&\min|{\lambda}({\bm N}_{\rm gfl}(j\omega))-{\lambda}({\bm N}_{\rm gfm}(j\omega))|\\
&\geq |\min {\lambda}({\bm N}_{\rm gfl}(j\omega))|-|\max {\lambda}({\bm N}_{\rm gfl}(j\omega))|\\
&\geq \left|\underset{\forall i,j}{\min}\left|\frac{J_i\omega^2}{\omega_0v_{d,i}}-\lambda_j(L_4)\right|+j\underset{\forall i}{\min}\frac{D_i\omega}{\omega_0v_{d,i}}\right|\\
&-\left|\underset{\forall i,j}{\max}\left|\frac{\omega^2}{k_{I,i}}
-\lambda_j(L_1)\right|-j\underset{\forall i}{\max}\frac{k_{P,i}\omega}{k_{I,i}}\right|=\delta\,.
\end{aligned}
\end{equation}
Considering $v_{d,i}\approx 1$pu at small-signal, then \eqref{eq:decoupled condition}
in Proposition~\ref{The decoupling condition} holds.}

\textcolor{catalogueblue}{Further, $\delta \geq {\min}\frac{D_i\omega}{\omega_0}-{\max}\frac{k_{P,i}\omega}{k_{I,i}}$, substituting it into \eqref{eq:decoupled condition} yields \eqref{eq:simple condition}. }

 \section{\textcolor{catalogueblue}{The proof of Proposition~\ref{stability condition}}}
 \label{sec:appendixC}
 
\begin{proof}
According to \eqref{eq:decoupled}, the equivalent characteristic polynomial can be rewritten as
\begin{equation}\label{eq:GNC function}
\begin{aligned}
\begin{cases}
 {\rm det}\left\{{\bm I}_n+[{\bm D}_{\rm gfl}s+{\bm L}_1]\frac{{\bm H}^{-1}_{\rm gfl}}{s^2}\right\}=0\\
 {\rm det}\left\{{\bm I}_m+[{\bm D}_{\rm gfm}(s)s+{\bm L}_4(s)]\frac{{\bm H}^{-1}_{\rm gfm}}{s^2} \right\}=0
\end{cases}
\,.
\end{aligned}
\end{equation}

Because ${\bm H}^{-1}_{\rm gfl}/s^2$ and ${\bm H}^{-1}_{\rm gfm}/s^2$ are diagonal matrices and their eigenvalues phase are all $-180^\circ$, according to GNC, if at the crossing frequency $s=j\omega_c$, the phases of eigenvalues of ${\bm D}_{\rm gfl}s+{\bm L}_1$ and ${\bm D}_{\rm gfm}(s)s+{\bm L}_4(s)$ are greater than $0^\circ$, the system is stable.

The conditions can be transformed to the real part of eigenvalues of  ${\bm D}_{\rm gfl}+{\bm L}_1/s$ and ${\bm D}_{\rm gfm}(s)+{\bm L}_4(s)/s$ are lager than $0$, i.e. at across frequency $s=j\omega_c$,
\begin{equation}\label{eq:positive realness}
\begin{aligned}
\begin{cases}
\underline{\lambda}\left[\frac{{\bm D}_{\rm gfl}+{\bm D}^H_{\rm gfl}}{2}+\frac{{\bm L}_1-{\bm L}^H_1}{2s}\right]>0\\
\underline{\lambda}\left[\frac{{\bm D}_{\rm gfm}(s)+{\bm D}^H_{\rm gfm}(s)}{2}+\frac{{\bm L}_4(s)-{\bm L}^H_4(s)}{2s}\right]>0
\end{cases}
\,,
\end{aligned}
\end{equation}
where $(\cdot)^H$ represents the conjugate transpose. For any matrix $A$, if $\underline{\lambda}(A+A^H)>0$, then the real part of the eigenvalues of $A$ is greater than $0$.

Under the decoupled conditions with small $\varphi$, as shown in Fig. \ref{fig9}, ${\bm K}^{\rm sy}_{\rm gfl}\approx0$ and ${\bm L}_1$ is approximately a diagonal matrix, thus $\frac{{\bm L}_1-{\bm L}^H_1}{s} \approx 0$. And ${\bm D}_{\rm gfl}\approx{\rm diag}\left\{T_{P,i}\right\}{\bm V}^{\rm gfl}_d+{\bm K}^{\rm d}_{\rm gfl}$ is approximately a Hermitian matrix, thus $\frac{{\bm D}_{\rm gfl}+{\bm D}^H_{\rm gfl}}{2}\approx {\bm D}_{\rm gfl}$. The condition is simplified as $\underline{\lambda}({\bm D}_{\rm gfl})\approx\underline{\lambda}({\rm diag}\left\{T_{P,i}v_{d,i}\right\}+{\bm K}^{\rm d}_{\rm gfl})\geq\underset{\forall i\in[1,n]}{{\rm min}}(T_{P,i}v_{d,i}) +\underline{\lambda}({\bm K}^{\rm d}_{\rm gfl})>0$, because ${\rm diag}\left\{T_{P,i}v_{d,i}\right\}$ and ${\bm K}^{\rm d}_{\rm gfl}$ are Hermitian matrices that satisfy the eigenvalue interlacing theorem ($T_P=k_P/k_I$). Thus, we obtain the stability conditions shown in the first function of \eqref{eq:static stability} from \eqref{eq:positive realness}.

For the GFM subsystem, ${\bm D}_{\rm gfm}(s)={\rm diag}\left\{T_{D,i}\right\}+\alpha(s){\bm K}^{\rm d}_{\rm gfm}$, ${\rm diag}\left\{T_{D,i}\right\}$ is a Hermitian matrix ($T_D=D$) and ${\bm K}^{\rm d}_{\rm gfm}$ is a skew-Hermitian matrix. Thus, $\frac{{\bm D}_{\rm gfm}(s)+{\bm D}^H_{\rm gfm}(s)}{2}={\rm diag}\left\{T_{D,i}\right\}+{\alpha^*(s)}{\bm K}^{\rm d}_{\rm gfm} $. In addition, ${\bm L}_4(s)=\alpha(s){\bm N}^{\rm sy}_{\rm gfm}+{\bm I}^{\rm gfm}_d$, ${\bm N}^{\rm sy}_{\rm gfm}$ and ${\bm I}^{\rm gfm}_d$ are Hermitian matrices. Thus, ${\bm L}_4(s)=\frac{{\bm L}_4(s)+{\bm L}^H_4(s)}{2s}={\alpha^*(s)}\frac{{\bm N}^{\rm sy}_{\rm gfm}}{s}$. According to the eigenvalue interlacing theorem, then we obtain the stability conditions shown in the second equation of \eqref{eq:static stability} from \eqref{eq:positive realness}.
\end{proof}

 \section{The equivalent GFL subsystem model}
 \label{sec:appendixD}
\textcolor{catalogueblue}{Multiplying the first equation in \eqref{eq:Shcur} by ${\rm diag}(T_{J,i})s^2,i\in [1,n]$ on the left gives 
 \begin{equation}\label{eq:eqgfl}
 \begin{aligned}
 &{\rm diag}(T_{J,i})s^2+ {\rm diag}(T_{P,i}s+1)\left({\bm L}_1-\widetilde{{\bm L}}(s)\right)=\\
  &{\rm diag}(T_{J,i})s^2+0.1 {\rm diag}(T_{P,i}s+1){\bm V}^{\rm gfl}_d+\\
  & {\rm diag}(T_{P,i}s+1)\left({\bm L}_1-\widetilde{{\bm L}}(s)-0.1{\bm V}^{\rm gfl}_d\right)
 \,.
 \end{aligned}
 \end{equation}
where, $\widetilde{{\bm L}}(s)=-{\bm K}^d_{\rm gfl}s+{\bm L}_2{\bm N}^{-1}_{\rm gfm}(s){\bm L}_3$.}

\textcolor{catalogueblue}{Left-multiply \eqref{eq:eqgfl} by matrix $\left[{\rm diag}(T_{J,i})s^2+0.1 {\rm diag}(T_{P,i}s+1){\bm V}^{\rm gfl}_d\right]^{-1}$, then we can obtain
 \begin{equation}\label{eq:eqgfl1}
 \begin{aligned}
 &{\bm I}_n+{\rm diag}\left(\frac{T_{P,i}s+1}{T_{J,i}s^2+0.1T_{P,i}{\bm V}^{\rm gfl}_{d,i}s+0.1{\bm V}^{\rm gfl}_{d,i}}\right)
 \\&\left({\bm L}_1-\widetilde{{\bm L}}(s)-0.1{\bm V}^{\rm gfl}_d\right)
 \,,
 \end{aligned}
 \end{equation}
which is same as \eqref{eq:small phase theroy function}.}

 \section{The parameters of converters}
 \label{sec:appendixE}
 \begin{table}[h]
\centering
\caption{Parameters of the Three-Converter System}
\renewcommand{\arraystretch}{1.1}
\scriptsize 
\begin{tabular}{@{}llll@{}} 
\toprule
\multicolumn{4}{c}{\textbf{Base Values for Per-unit Calculation}} \\
\midrule
$f_{\text{base}} = 50\,\text{Hz}$ & $\omega_{\text{base}} = 2\pi f_{\text{base}}$ & $U_{\text{base}} = 220\,\text{kV}$ & $S_{\text{base}} = 100\,\text{MVA}$\\
\midrule
\multicolumn{4}{c}{\textbf{Parameters of the Converter (p.u.)}} \\
\midrule
$L_f = 0.05$& $C_f = 0.05$& ${\rm PI}_{\rm C}(s)=1+2/s$ &  $T_v=0.001$\\
\midrule
\multicolumn{4}{c}{\textbf{Parameters of the Grid-Forming Converter (p.u.)}} \\
\midrule
$J = 2$ & $D = 20$ &  ${\rm PI}_{\rm V}(s)=0.4+1.5/s$ &  \\
\midrule
\multicolumn{4}{c}{\textbf{Parameters of the Grid-Following Converter (p.u.)}} \\
\midrule
$k_P = 15$ & $k_I = 3500$ &   &  \\
\bottomrule
\end{tabular}
\end{table}
}

\bibliographystyle{IEEEtran}
\bibliography{IEEEabrv,RS}

\end{document}